\documentclass[twoside]{article}

\usepackage{placeins}
\usepackage{svg}
\usepackage{lipsum}
\usepackage{enumitem}
\usepackage[labelformat=parens,labelsep=space]{subcaption}
\usepackage{natbib}
\usepackage{snaptodo}

\snaptodoset{block rise=0.5em}
\snaptodoset{margin block/.style={font=\tiny}}
\PassOptionsToPackage{numbers, compress}{natbib}
\usepackage[utf8]{inputenc} 
\usepackage[T1]{fontenc}    
\usepackage{hyperref}       
\usepackage{url}            
\usepackage{booktabs}       
\usepackage{amsfonts}       
\usepackage{nicefrac}       
\usepackage{microtype}      
\usepackage{xcolor}         
\colorlet{darkgreen}{green!45!black}
\hypersetup{
  colorlinks,
  citecolor=darkgreen}
\usepackage{hyphenat}

\usepackage{adjustbox}
\usepackage{amsfonts}
\usepackage{xcolor}

\usepackage{pifont}

\usepackage{graphicx}
\usepackage{enumitem}

\usepackage{amsmath, amssymb, mathtools, amsthm}
\usepackage{thmtools, thm-restate}

\theoremstyle{plain}
\declaretheorem{theorem}
\declaretheorem[sibling=theorem, numberwithin=section]{proposition}
\declaretheorem[sibling=theorem, numberwithin=section]{lemma}
\declaretheorem[sibling=theorem, numberwithin=section]{corollary}
\theoremstyle{definition}
\declaretheorem[sibling=theorem, numberwithin=section]{definition}
\declaretheorem[sibling=theorem, numberwithin=section]{assumption}
\theoremstyle{remark}

\newcommand{\defeq}{\overset{\text{\tiny def}}{=}}
\newcommand{\AJ}{\rm AJ}
\renewcommand{\d}{\rm d}
\renewcommand{\P}{\mathbb{P}}
\newcommand{\E}{\mathbb{E}}
\newcommand{\D}{\mathcal{D}}

\newcommand{\softmax}{\rm softmax}
\newcommand\indep{\protect\mathpalette{\protect\indepT}{\perp}}
\def\indepT#1#2{\mathrel{\rlap{$#1#2$}\mkern2mu{#1#2}}}

\usepackage{bbold}
\usepackage{bigints}
\usepackage{stmaryrd}
\usepackage{booktabs,tabularx}
\usepackage{nicefrac, xfrac}
\usepackage{makecell}
\usepackage[nobiblatex]{xurl}
\newcommand\rurl[1]{%
  \href{http://#1}{\nolinkurl{#1}}%
}

\DeclareMathOperator*{\argmin}{arg\,min}

\usepackage{stmaryrd}

\usepackage{tikz}
\usetikzlibrary{backgrounds}
\usetikzlibrary{tikzmark}
\usetikzlibrary{calc}
\usetikzlibrary{arrows,shapes,positioning,shadows,trees,mindmap}
\usetikzlibrary{arrows.meta}

\usepackage{tcolorbox}
\newcommand{\highlight}[2]{\colorbox{#1!17}{$\displaystyle #2$}}

\renewcommand{\highlight}[2]{\colorbox{#1!17}{#2}}

\usepackage[accepted]{aistats2026}

\usepackage{algorithm}
\usepackage{algpseudocode}
\algtext*{EndWhile}
\algtext*{EndIf}
\algtext*{EndFor}
\algtext*{EndLoop}

\usepackage{geometry}
\begin{document}

\addtocontents{toc}{\protect\setcounter{tocdepth}{0}}
\twocolumn[

\aistatstitle{On the calibration of survival models with competing risks}
\aistatsauthor{Julie Alberge \And Tristan Haugomat \And Gaël Varoquaux \And Judith Abécassis}
\aistatsaddress{Inria de Saclay \And Inria de Saclay, Drees \And Inria de Saclay, :probabl. \And Inria de Saclay}
]

\begin{abstract}
    Survival analysis deals with modeling the time until an event occurs, and accurate probability estimates are crucial for decision-making, particularly in the competing-risks setting where multiple events are possible. While recent work has addressed calibration in standard survival analysis, the competing-risks setting remains under-explored as it is harder (the calibration applies to both probabilities across classes and time horizon). We show that existing calibration measures are not suited to the competing-risk setting and that recent models do not give well-behaved probabilities. To address this, we introduce a dedicated framework with two novel calibration measures that are minimized for oracle estimators (\emph{i.e.}, both measures are proper). We also introduce some methods to estimate, test, and correct the calibration. Our recalibration methods yield good probabilities while preserving discrimination.
\end{abstract}

\section{INTRODUCTION: \emph{FAITHFUL} PROBABILITIES MATTER}

\textit{Is my patient more likely to die of a heart attack or  breast cancer in the coming year?} In various application domains, including health or marketing, accurate probabilities are crucial. 
Competing risks analysis focuses on modeling several mutually exclusive \textit{time-to-event outcomes}. More specifically, one wants to predict which of the events of interest will happen first, and when. An application could be distinguishing between different causes of death. One challenge is right censoring: certain events are not observed. 
Practically, competing risks methods estimate the cumulative distribution functions of each event of interest for each individual, for example, the probabilities of dying from each cause at any time for a given patient. 
Initially, simple estimators were proposed, such as \citet{aalen_empirical_1978} (see Definition in Appendix \ref{def:aj}), a marginally consistent estimator that returns the incidence function for each event type for the entire population but does not predict individualized cumulative incidence functions.
Recent research has developed more and more complex models using covariates to individualize the predictions of the cumulative incidence functions for each patient \citep{fine_proportional_1999, ishwaran_random_2008, lee_deephit_2018, wang_survtrace_2022}. Those models estimate conditional probabilities that can be integrated over the whole population to obtain a marginal prediction. Some of these models, in particular, based on machine learning approaches like forests~\citep{ishwaran_random_2008} or deep neural networks~\citep{lee_deephit_2018, wang_survtrace_2022} exhibit excellent discriminative performance, but their calibration is not studied. \\
Indeed, while discrimination between patients is useful, probabilities are key to decision making \citep{VanCalster2019achilles,perez2025decision}, \emph{e.g.}, to arbitrage between treatments depending on the probabilities of different causes of death. But to be useful, these probabilities must be \emph{faithful} in some sense. \emph{Calibration} is one way to assess their faithfulness, comparing the average individualized predicted probabilities with population-level predictions from a marginally consistent estimator \cite[\emph{e.g.}][]{aalen_empirical_1978}. On the METABRIC dataset, following various causes of death~\citep{rueda2019dynamics}, this comparison reveals that existing methods produce predictions that, on average, deviate from the marginally consistent estimator, suggesting poor calibration (Figure~\ref{fig:meanincfunc}). 

\begin{figure*}[t]
    \hfill%
    \includegraphics[height=.29\textwidth]{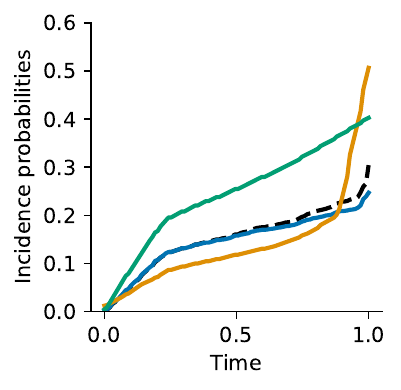}%
	\llap{\raisebox{.26\linewidth}{Event 1}\qquad}%
    \includegraphics[height=.28\textwidth]{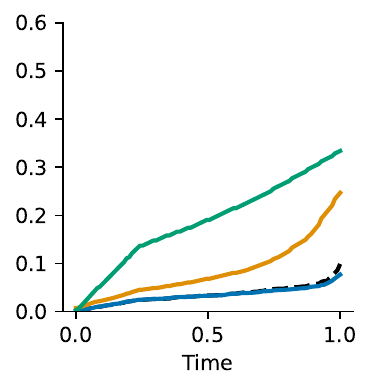}%
	\llap{\raisebox{.26\linewidth}{Event 2}\qquad}%
    \includegraphics[height=.28\textwidth]{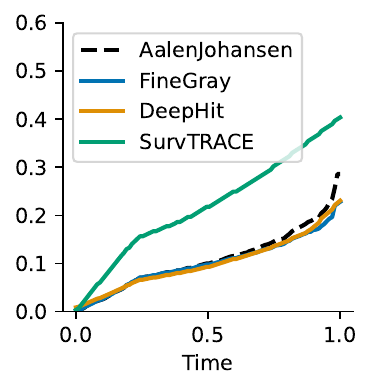}%
	\llap{\raisebox{.26\linewidth}{Event 3}\qquad}%
    \hfill\vbox{}%
    \caption{\textbf{Marginal probabilities of SOTA models do not behave as expected:} We compare the marginal probabilities of the different SOTA models on the SEER Dataset with 10k training samples. All these marginals should be close to the \citet{aalen_empirical_1978} estimator, which is marginally consistent.\label{fig:meanincfunc}}
\end{figure*}

For survival analysis, a simpler setting with a single event of interest, different measures of calibration have recently been studied \citep{haider_eective_nodate, qi_conformalized_2024}. 1-calibration focuses on predictions at a given time point to recover the usual binary setting \citep{haider_eective_nodate}; it can be extended to time-to-event regression by integrating over time, leading to the KM-calibration~\citep{qi_conformalized_2024}. D-calibration focuses on the properties of the cumulative distribution function of the survival function \citep{haider_eective_nodate}.\\

Survival-analysis calibration cannot be directly applied to competing-risks setting (see \autoref{fig:recalsurvival} and \autoref{tab:calibscospe}); the multiclass aspect needs to be taken into account. Practitioners in need of ad-hoc solutions have extended 1-calibration to calibration plots (checking the calibration of different classes at one given time point) \citep{huang_tutorial_2020}. But work is needed to define an actually \emph{proper} metric that addresses both the multiple classes and the whole period of prediction.
The challenge is indeed integrating in both directions (events and time), as illustrated in \autoref{fig:sotacalibs}. \\
This paper defines proper measures of calibration for competing risks, as well as estimators for those measures, by combining the notion of multi-class calibration and survival analysis calibration. 

\begin{table}[t]
    \centering
    \caption{\textbf{Inadequacy of survival-only calibration scores for competing risks models.} On METABRIC (10 seeds), survival calibration metrics fail to capture the consistency of the \citet{aalen_empirical_1978} estimator. While D-calibration \citep{haider_eective_nodate} suggests calibration for event 1 (death from breast cancer) but not event 2 (death from other causes), KM-calibration \citep{qi_conformalized_2024} incorrectly indicates poor calibration for both.
    }
    \begin{tabular}{lrr}
        \toprule
        Event & D-calibration & KM-calibration \\
        & (mean±std) & (mean±std) \\
        \midrule
        1 (breast-cancer)& 1.0±.0 & .0012±.0011 \\
        2 (other)& 0.2±.4 & .0393±.0113 \\
        \bottomrule
    \end{tabular}
    \label{tab:calibscospe}
\end{table}

\begin{figure}[t]
    \centering
    \includegraphics[width=\linewidth]{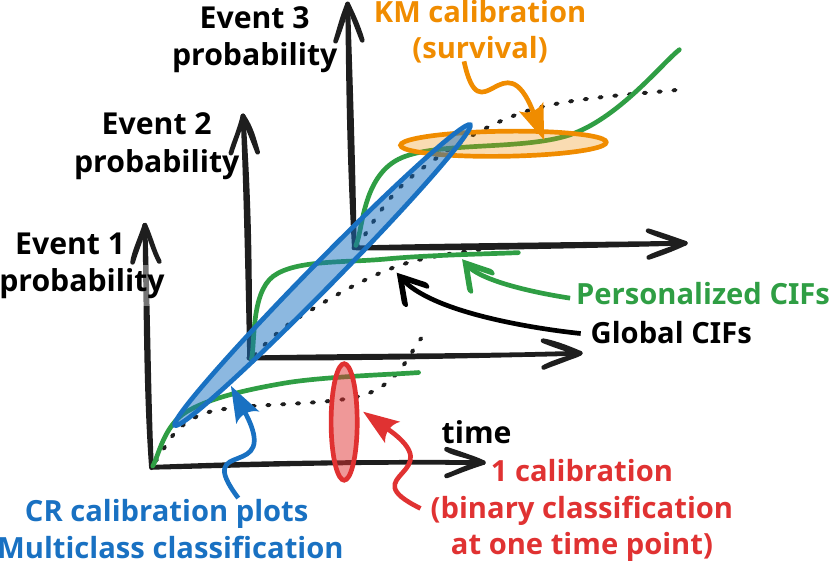}
    \caption{\textbf{Existing calibration measures cast to the competing-risks setting:} The CIF (Cumulative Incidence Function) for an individual specifies, for each event type, the probability of experiencing that event before time $\tau$. In this space, different existing calibration measures correspond to projecting the vector-valued CIF onto a specific "direction": 1-calibration (binary calibration at a fixed time), KM-calibration (survival/Kaplan–Meier calibration), multiclass CR calibration plots, or calibration of global (population-level) CIFs. Each of these captures only one marginal aspect of the problem and fails to assess the joint calibration of the full CIF vector across event types.}
    \label{fig:sotacalibs}
\end{figure}

\paragraph{Contributions}
After highlighting that state-of-the-art (SOTA) competing-risks methods do not return marginal meaningful probabilities, we contribute a full calibration framework for competing risks:
\begin{itemize}[leftmargin=1.7ex]
    \item We extend D-calibration to competing-risks settings, proposing the CR D-calibration. We prove that this metric is proper: the oracle functions are CR D-calibrated. We also prove that any marginally consistent estimator is asymptotically CR D-calibrated, assuming the strict monotonicity of the cumulative incidence functions. We propose a consistent estimator of this metric and a statistical test. 
    \item We introduce the ${\rm cal}_{K}^\alpha$-calibration metric that measures the difference between the marginal oracle function at any time point $\tau$ (\emph{i.e.} $F_k(\tau)$) with the probabilities obtained by the estimator integrated over the samples. We prove that this metric well behaved (being equals to 0 for the oracle functions and calibrated ones), give a consistent estimator and propose a test.
    \item We propose two post-hoc re-calibration methods for competing risks models. We show that they maintain the same discrimination (\emph{i.e.}, unchanged C-index), while also making models more calibrated. We study the impact of the recalibration on other competing risks performance metrics. 
\end{itemize}

\section{BACKGROUND}

\subsection{Problem formulation: competing-risks setting}
We consider $K \in \mathbb{N}^*$ mutually exclusive competing events. For $k \in \llbracket 1, K \rrbracket$, we denote $T^*_k \in \mathbb{R}_+$ the event time of the event $k$, depending on the covariates $ \mathbf{X}$. We also denote $T^* \in \mathbb{R}_+$, the time of the first event that occurs, $T^* = \min\limits_{k \in \llbracket 1, K \rrbracket}(T^*_k)$ and $\Delta^* = \argmin\limits_{k \in \llbracket 1, K \rrbracket}(T^*_k)$.
We observe $(\mathbf{X}, T, \Delta) \sim \mu$, where $\mu$ denotes the distribution followed by $(\mathbf{X}, T, \Delta)$, with $T = \min(T^*, C)$ where $C \in \mathbb{R}_+$ is the censoring time, which can depend on $\mathbf{X}$, and $\Delta \in \llbracket0, K \rrbracket, \Delta= \argmin\limits_{k \in \llbracket 0, K \rrbracket}(T^*_k) $, where 0 denotes a censored observation. 
However, we are primarily interested in the distribution of the uncensored data, $(\mathbf{X}, T^*, \Delta^*) \sim \mu^*$, particularly the conditional distribution of $T^*, \Delta^*|\mathbf{X} =\mathbf{x}$. We index the observations by $i$, and $\textbf{x}_i$ denotes the covariates associated with the $i^{th}$ observation. The outcome is represented by $(t_i, \delta_i)$, where $t_i$ is the observed time, and $\delta_i \in \llbracket0, K\rrbracket $ is the event indicator. For $k \in \llbracket 1, K \rrbracket$, $\delta_i = k$ indicates that the event of interest $k$ was observed at time $t_i$, while $\delta_i = 0$ indicates that the observation was censored at time $t_i$. 
We divide our data set $\mathcal{D} = (\mathcal{X}, \mathcal{Y})$ into a training set $\mathcal{D}_{train}$, a calibration data set $\D_{cal}$, and $\mathcal{D}_{test}$. We introduce the different quantities estimated by competing risks models: 

\begin{definition}[\emph{Quantities of interest}]\label{quantities_interest}
\begin{flalign*}%
    \text{\parbox{\linewidth}{\footnotesize{CIF} \small (Cumulative Incidence
    Function):}}\\[-.1em] F^*(\tau|\mathbf{x}) = \mathbb{P}(T^* \leq \tau
|\mathbf{X}= \mathbf{x}) \\[.1em]
    \text{\parbox{\linewidth}{\footnotesize{CIF of the $k^{th}$ event}:}}\\ F^*_k(\tau|\mathbf{x}) = \mathbb{P}(T^* \leq \tau \cap \Delta = k |\mathbf{X}= \mathbf{x}) \\
    \text{\parbox{\linewidth}{\footnotesize{Survival Function to any event}:}}\\[-.2em] 
    S^*(\tau|\mathbf{x}) = \mathbb{P}(T^* > \tau |\mathbf{X}= \mathbf{x})
\end{flalign*}%
\end{definition}%

In the remainder of this work, we will use the following assumptions: 
\begin{assumption}[\emph{Continuity of the CIFs and Survival Function}]\label{ass:continous}
    We assume that each cumulative incidence function $F_k$ and the survival function are continuous.
\end{assumption}

\begin{assumption}[\emph{Non-informative censoring}]\label{ass:noninfocensoring}
    We make the classic assumption in survival analysis or in the competing-risks setting that the censoring is non-informative with respect to covariates: $$ T^* \indep C \, | \, \mathbf{X}$$
\end{assumption}

\begin{assumption}[\emph{Independence}]\label{ass:exchangeability} 
    We suppose that the samples are drawn i.i.d.
\end{assumption}

\begin{assumption}[\emph{Each event may occur eventually}]\label{ass:nonnul}
    For each event of interest, we make the minor assumption that the event may happen to each individual one day, \emph{i.e.} $\P(\Delta^*=k | \mathbf{X}) \neq 0$. 
\end{assumption}

While in survival analysis, one assumes that the event will occur one day, \emph{i.e.} $\mathbb{P}(T^* < \infty |\mathbf{X}) =1$, in the competing-risks setting, we make an equivalent assumption here, which will only be necessary in the definition of the CR-D-calibration (Def. \ref{def:dcal}):
\begin{assumption}[\emph{One event will happen a.s.}]\label{ass:dieatinfty}
    For all individuals, we assume that one event will happen almost surely one day,  for a sufficiently long time, \emph{i.e.} $\sum_{k=1}^K \mathbb{P}(T^* < \infty, \Delta^* = k |\mathbf{X}) = 1$. 
\end{assumption}

\subsection{Related work} \label{related_work}

\paragraph{The importance of competing risks}
The competing-risks setting accounts for multiple, mutually exclusive event types in time-to-event studies. For example, consider estimating the number of breast cancer-related deaths in a population. Focusing on one possible cause of death while ignoring other potential causes leads to biased estimates \citep{gorfine_frailty-based_2011, wolbers_prognostic_2009}.While the competing-risks setting is more complex than typical survival analysis, various approaches exist to tackle this scenario \citep{aalen_empirical_1978, fine_proportional_1999, lee_deephit_2018,  VanCalster2019achilles, nagpal,alberge_survival_2025}.
 
\paragraph{Calibration metrics in survival analysis settings}  
Introduced by~\citet{haider_eective_nodate}, D-calibration in the survival analysis setting is based on a property of the cumulative distribution function. They assume that the event of interest will happen almost surely, \emph{i.e.} $\forall x, F(\infty|\mathbf{x}) = \mathbb{P}(T^* < \infty \mid \mathbf{X} = \mathbf{x}) = 1$, where $F$ is the cumulative distribution function of $T^*$. By the probability integral transform \citep{david_1948, wikith}, it follows that $F(T^*) \sim \mathcal{U}(0,1)$ (as detailed in~\citet{casella}. \citet{haider_eective_nodate} extend this idea and incorporate censored individuals, given that the patients have survived up to their censoring time. A formal definition of the D-calibration in the survival analysis setting can be found in Appendix \ref{def:dcalsurvival}.
The other main calibration metric in survival analysis is the 1-calibration \citep{huang_tutorial_2020}, which evaluates the model calibration at a fixed time point by considering the problem as a binary classification task. It comes with its extension: the \citet{kaplan_nonparametric_1958} calibration \citep{chapfuwa_calibration_2023}.

\paragraph{Calibration in multi-class settings} In classification without censoring, multi-class calibration can be characterized in several ways, some stronger than others \citep{ kull_beyond_2019, perez-lebel_beyond_2023, vaicenavicius_evaluating_2019}. Top-label calibration \citep{guo_calibration_2017} considers that a model is calibrated if the highest predicted class corresponds to the actual class. Class-wise calibration, more stringent, controls the probability of each class \citep{zadrozny_transforming_2002}. Finally, joint calibration evaluates the calibration of the entire prediction vector across all classes simultaneously \citep{kull_beyond_2019}. Post-hoc methods can recalibrate a multi-class classifier \citep{filho_classifier_2023}. Isotonic regression \citep{zadrozny_transforming_2002} is a classic technique, but it does not control for multiclass probabilities. These are often cobbled together using one-versus-rest \citep{berta_rethinking_2025, isotonic_scikit-learn}. Temperature scaling, using a softmax parameterization, is better behaved \citep{guo_calibration_2017}.

\paragraph{Calibration for competing risks models} The notion of calibration for competing risks lacks theoretical grounding. Nevertheless, as \citet{wolbers_prognostic_2009} show in their review, most works tackle it using calibration plots at a given fixed time point using pseudo-observations to handle censoring, or with a marginal estimator \citep[\emph{e.g.}][]{aalen_empirical_1978,zhang_overview_2018}. The calibration plot can also be constructed using more complex methods \citep{austin_graphical_2022, gerds_calibration_2014}. \citet{booth_using_2023} used shrinkage to recalibrate competing risks methods as a post-hoc operation. But, to our knowledge, there is no real-valued global calibration metric in the competing-risks setting that has been formally defined and that addresses the multiple dimensions of the problem together, as illustrated graphically in Figure~\ref{fig:sotacalibs}.

\section{MEASURING CALIBRATION IN COMPETING-RISKS SETTINGS}
\subsection{Distribution calibration in the competing-risks setting (CR D-cal)}

\begin{figure}[b!]
    \includegraphics[width=\linewidth]{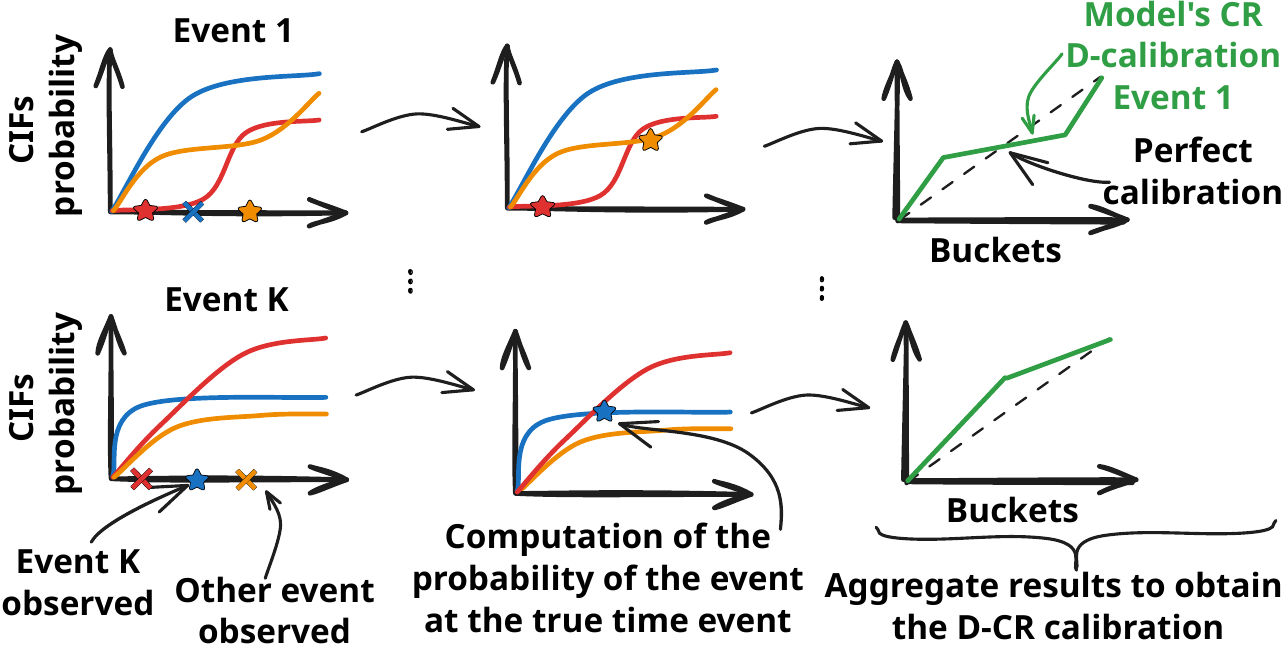}
    \caption{\textbf{CR D-calibration without censoring:} For each event, we compute the individualized CIFs. Then, for each observed individual $\mathbf{x}_i$ with event $k$, we compute $F_k(t_i|\mathbf{x}_i)$ and $F_k(\infty|\mathbf{x}_i)$. Following theorem \ref{thm:bigth}, we compare the $\nicefrac{F_k(t_i|\mathbf{x}_i)}{F_k(\infty|\mathbf{x}_i)}$ cumulative empirical distribution to an uniform cumulative distribution function. Summing this metric over all event types the results gives a measure of the CR D-calibration.%
    \label{fig:dcal-expla}}%
\end{figure}

\paragraph{Without censoring} For now, we suppose there are no censored individuals. Importantly, the D-calibration for survival analysis introduced in section~\ref{related_work} does not easily carry over to the competing-risks setting. Indeed, in the competing-risks setting, a patient does not have a 100\% chance of contracting one particular event and not the other (Assumption \ref{ass:nonnul}).
As a result, the key property for D-calibration in survival analysis \emph{i.e.} $\forall x, F(\infty|\mathbf{x}) = \mathbb{P}(T^* < \infty \mid \mathbf{X} = \mathbf{x}) = 1$ does not extend to competing risks: for each event $k$, $F_k(T^*, \mathbf{X}) \nsim \mathcal{U}(0, 1)$ because $F_k(\infty |\mathbf{X}) = \mathbb{P}(T^* < \infty, \Delta^* = k \mid \mathbf{X} = \mathbf{x}) = \mathbb{P}(\Delta^* = k | \mathbf{X} = \mathbf{x}) \neq 1$.
But as a direct extension, we see almost surely that $\forall k \in \llbracket 1, K \rrbracket, \, F_k(T^* |\Delta^*=k, \mathbf{X}) \, \sim \, \mathcal{U}(0, \mathbb{P}(\Delta^* = k | \mathbf{X}))$. However, this uniformity cannot be generalized to the population level, but an easier way to see this quantity is to take the ratio $\nicefrac{F_k(T^*|\mathbf{X})}{F_k(\infty|\mathbf{X})}$ which simply follows a uniform distribution between 0 and 1, \emph{i.e.}, $\left(\nicefrac{F_k(T^*|\mathbf{X})}{F_k(\infty|\mathbf{X})} | \Delta^* =k, \mathbf{X}\right) \, \, \sim \, \, \mathcal{U}(0, 1)$ (property of a cumulative distribution function, see Lemma \ref{lemm:ratiocdf}, Appendix \ref{dcalsurvival}). Conditioning on the event of interest and taking the ratio of the cumulative incidence functions allows us to use the probability integral transform and define the CR-D-calibration in the case where there is no censoring (illustrated on fig.\,\ref{fig:dcal-expla}). In the following, we will write $F_{k/\infty}(t \mid \mathbf{x}) \defeq \frac{F_k(t \mid \mathbf{x})}{F_k(\infty \mid \mathbf{x})}.$

\paragraph{Incorporating the censoring mechanism} Competing risks are concerned with censored individuals. 
Incorporating censored samples is more subtle than in the single-event case. For a censored individual, the potential event type is unobserved, and the subject will experience \textit{exactly one} of the competing events. Hence, unlike standard survival analysis, we cannot directly impute the conditional distribution of $T^*$ given $\mathbf{X}$ and $\Delta = 0$, as the terminal event remains unknown. Nonetheless, under the assumption that each individual will experience an event (\emph{i.e.}, $\sum_k F_k(\infty \mid \mathbf{X}) = 1$) (Assumption \ref{ass:dieatinfty}), and that the individual has survived up to the censoring time $C$, we define the competing risks D-calibration (CR D-calibration) by assigning, for each cause $k$, the distribution of $\mathbb{P}(T^* \leq \tau, \Delta^* = k | \mathbf{X} , T^* > C) = \nicefrac{\P(C\leq T^\star \leq \tau, \Delta^\star | \mathbb{X})}{\P(T^\star > C)}$ (Definition \ref{def:dcal}) hence attributing a contribution of each censored observation to all possible events following their relative probability distribution. This construction enables accounting for censored observations by distributing each event-specific risk beyond $C$ in a way that respects the model’s marginal predictions. With this added term, we recover a uniform distribution (see Theorem \ref{thm:bigth}). To compute the CR-D-calibration, we construct buckets (to compare them with those from a uniform distribution), $B_{[0, \rho]}$ with $\rho \in [0, 1]$. For a given bucket, we compute two terms. In the first one, we measure the patients who have experienced the event of interest $\mathbb{1}_{\hat{F}_{k/\infty}(T \mid \mathbf{X}) \in [0,\rho], \ \Delta = k}$. The second term represents the patient that have been censored ($\mathbb{1}_{C \leq T^*, \ \hat{F}_{k/\infty}(C\mid \mathbf{X}) \in [0,\rho]}$), this adjustment term is based on the expected value given their known censoring time:
$\nicefrac{\hat{F}_k(\infty\mid \mathbf{X})\rho - \hat{F}_k(C\mid \mathbf{X})}{\hat{S}(C\mid \mathbf{X})}$.
\begin{definition}[CR D-calibration]\label{def:dcal}
     For any estimator of the CIFs, we measure the following quantity, for a given $\alpha \in \mathbb{R}_+, \alpha > 1$ where $\alpha$ represents the $\alpha$-norm. 
    The numerator estimates the measure of the bucket, and the denominator normalizes by the total predicted mass of event (since some individuals may never experience k).
    \begin{flalign*}
        \text{\footnotesize{CR D-cal}} &&  \!\!D^{CR}_{\alpha} =\!  \sum_{k=1}^K  \!\bigg(\int_{0}^1 \!\Big|\! \underbrace{\frac{\mathbb{E}_{T, \Delta}(B^k_{[0, \rho ]})}{\E_X(\hat{F}_k(\infty |X))}}_{\defeq b^k_{[0, \rho]}} -  \rho \Big|^\alpha \!d\rho \bigg)^{\!\frac{1}{\alpha}}&&
    \end{flalign*}%
    where\; $\forall \rho \in [0, 1]$
    \begin{multline*}
        B^k_{[0, \rho]} = \tikzmarknode{observed}{\highlight{violet!50!white}{$\!\mathbb{1}_{\hat{F}_{k/\infty}(T\mid \mathbf{X})\in[0,\rho],~\Delta=k}\!$}} \\ + 
     \tikzmarknode{censored}{\highlight{green!80!white}{$\!\frac{\hat{F}_k(\infty\mid \mathbf{X})\rho - \hat{F}_k(C\mid \mathbf{X})}{\hat{S}(C\mid \mathbf{X})}\mathbb{1}_{C\leq T^*, \hat{F}_{k/\infty}(C\mid \mathbf{X})\in[0,\rho]}\!$}}
    \end{multline*} 
    \begin{tikzpicture}[overlay,remember picture,>=stealth,nodes={align=left,inner ysep=1pt},<-]
     \path (observed.north) ++ (4em, 2em) node[anchor=north,color=violet] (obs){\text{\parbox{20ex}{\small Observed samples}}};
     \draw [color=violet](observed.north) -- ([xshift=0.1ex,color=violet]obs.south);
    \path (censored.north) ++ (4.9em, 2.8em) node[anchor=north,color=green!50!black] (censor){\text{\parbox{16ex}{\small Probability after censoring}}};
     \draw [color=green!50!black](censored.north) -- ([xshift=-0.1ex,color=green!50!black]censor.south);
\end{tikzpicture} 
    A model is perfectly CR D-calibrated when $D_\alpha^{CR} = 0$which implies that the empirical distribution corresponds to the expected one. We can also define the intervals differently, between two reals $0\leq a < b\leq 1$ instead of $[0, \rho]$, $B^k_{[a,b]}=B^k_{[0,b]} - B^k_{[0,a]}$. This will add a new term to the construction of the interval $B_{[a, b]}$ detailed in Appendix \ref{proof:expec} (eq. \ref{eq:bucket_ab}).
\end{definition}

\begin{restatable}[Properness]{theorem}{bigth}\label{thm:bigth}
    Under assumptions \ref{ass:continous}, \ref{ass:noninfocensoring}, \ref{ass:exchangeability}, \ref{ass:nonnul}, the expectation of the bucket for the oracle function can be computed as: $\E(B_{[a,b]}\mid \mathbf{X}) = (b-a)F_k(\infty\mid \mathbf{X})$. 
    This implies that the oracle functions are CR D-calibrated (Def. \ref{def:dcal}) - \emph{i.e.}, $\forall \alpha \in \mathbb{R}^*, \, D^{CR}_\alpha$ (eq. 1) is minimized for the oracle CIFs. \\
    Under the same assumptions, a marginally consistent estimator is asymptotically CR D-calibrated. 
\end{restatable}

\begin{proof}[Proof sketch]
    We start by proving the stated expectation value (Lemma \ref{lemm:exporacle}) by studying the oracle function in those cases: if the patient is censored or if the event has been observed. Then, the result is straightforward with the Lemma \ref{lemm:exporacle}. 
    With Lebesgue's dominated convergence theorem, we show that a marginally consistent estimator is asymptotically CR D-calibrated.  For more details, see Appendix \ref{proof:thproperness}.
\end{proof}

\begin{corollary}
    As the Aalen-Johansen estimator is a consistent (see Lemma \ref{lemm:ajconsistent}) marginal estimator of the oracle functions, it is also asymptotically CR D-calibrated.
\end{corollary}
\paragraph{Estimator of the CR D-calibration}
 Extending Definition \ref{def:dcal} to the empirical measure, for a given population and a given model, we define the estimation of the CR D-calibration on given dataset (named $\hat{D}^{CR}_{\alpha}$). 

\begin{definition}[CR $\hat{D}^{CR}_\alpha$-calibration]\label{def:dcalestimation}
    Thus, we define, for a fixed $ \rho \in [0, 1]$, and event type $k$:
    \begin{multline*}
         \hat{b}^k_{[0, \rho]} = \frac{{\rm card}\left\{i ~\middle|~ \delta_i=k, \hat{F}_{k/\infty}(t_i\mid \mathbf{x}_i)\in [0, \rho]\right\}}{ \sum_{i}\{\hat F_k(\infty\mid \mathbf{x}_i)\}} \\
         + \frac{
            {\displaystyle\sum_{i}} \left\{\frac{\hat{F}_k(\infty\mid \mathbf{x}_i)\rho - \hat{F}_k(t_i\mid \mathbf{x}_i)}{\hat{S}(t_i\mid \mathbf{x}_i)} ~\middle|\!\begin{array}{l}
                 \small\delta_i=0  \\
                 \small\hat{F}_{k/\infty}(t_i\mid \mathbf{x}_i)\in[0,\rho])
            \end{array}\!\! \right\}
        }{
            {\sum_{i}}\{\hat F_k(\infty\mid \mathbf{x}_i)\}
        }
    \end{multline*}
    Where the first term counts individuals who experienced event $k$ and whose normalized prediction falls in $[0,\rho]$, scaled by the total predicted risk. The second term accounts for censored patients,  estimating their contribution. 
    Then, we can define an estimator of the CR $D^{CR}_\alpha$-calibration as: 
     \begin{flalign*}
        \hat{D}^{CR}_\alpha\text{-cal estimator} && \!\!\hat{D}^{CR}_\alpha \defeq \sum_{k=1}^K \! \left(\int_0^1\!\left| \hat{b}^k_{[0, \rho]} - \rho \right|^\alpha\! d\rho\right)^\frac{1}{\alpha}
\end{flalign*}
\end{definition}

\begin{proposition}[Consistency of the estimator]\label{prop:consisD}
    As the sample size ($m$) of the calibration set $\D_{cal}$ tends toward infinity (\emph{i.e.} $m \to \infty$), the estimator CR $\hat{D}$-calibration (Def. \ref{def:dcalestimation}) converges almost surely toward CR $D$-calibration (Def. \ref{def:dcal}), making it a consistent estimator of the CR D-calibration.  
\end{proposition}

\begin{proof}[Proof Sketch]
    The proof relies primarily on the strong law of large numbers, which ensures that the consistent estimator converges almost surely to its expected value, the CR-D-calibration.
\end{proof}

\begin{algorithm}[H]
\begin{algorithmic}
\Require $\{(t_i, \delta_i, \mathbf{x}_i)\}_{i=1}^n$, $\{\hat{F}_k\}_{k =1}^K$, $\hat{S}$, $\alpha$, $M$.

\State $\hat{D}^{CR}_\alpha \gets 0$

\For{$k \gets 1$ to $K$}
    \State $w_{total} \leftarrow \sum_{i=1}^N \hat{F}_k(\infty \mid \mathbf{x}_i)$ \textcolor{teal}{\# Denominator}
    
    \State $I_k \leftarrow 0$ \textcolor{teal}{\# Initialize integral}

    \For{$\rho \gets 1/M$ to $1$ with step $1/M$} 
        
        \State$Count_{k} \gets 0$
        \State $Sum_{cens} \gets 0$
        
        \For{patient $i \gets 1$ to $N$}
            \State \textcolor{teal}{\# Compute normalized prediction}
            \State $\hat{F}_{k/\infty} \gets \hat{F}_k(t_i \mid \mathbf{x}_i) / \hat{F}_k(\infty \mid \mathbf{x}_i)$

            \If{$\hat{F}_{k/\infty} \in [0, \rho]$}
                \If{$\delta_i = k$}
                    \State $Count_{k} \gets Count_{k} + 1$
                \EndIf
                \If{$\delta_i = 0$} 
                    \State \textcolor{teal}{\# Adjust for censored samples}
                    \State $cens \gets \frac{\left(\hat{F}_k(\infty \mid \mathbf{x}_i) \cdot \rho - \hat{F}_k(t_i \mid \mathbf{x}_i)\right)}{\hat{S}(t_i \mid \mathbf{x}_i)} $
                    \State $Sum_{cens} \gets Sum_{cens} + cens$
                    
                \EndIf
            \EndIf
        \EndFor
        \State \textcolor{teal}{\# Calculate bias for current $\rho$}
        \State $\hat{b}^k_{[0, \rho]} \gets \frac{Count_{k} + Sum_{cens}}{w_{total}}$
        
        \State $I_k \gets I_k +  \frac{\left| \hat{b}^k_{[0, \rho]} - \rho \right|^\alpha}{M}$
    \EndFor
    \State $\hat{D}^{CR}_\alpha \gets \hat{D}^{CR}_\alpha + (I_k)^{1/\alpha}$
\EndFor
\State \Return{$\hat{D}^{CR}_\alpha $}
\end{algorithmic}

\caption{Computation of $\hat{D}^{CR}_\alpha$-Calibration}
\end{algorithm}

\subsection{${\rm cal}_{K}^\alpha$-calibration, another calibration metric using marginal probabilities}

While CR D-calibration focuses on the properties of the cumulative distribution function, we can define calibration in another way. As illustrated in Figure \ref{fig:meanincfunc}, a desirable property of a model is to have globally faithful probabilities (\emph{i.e.}, $\forall \tau \in \mathbb{R}_+, \, \forall k\in \llbracket 1, K \rrbracket, \, \hat{F}_k(\tau) \approx F_k(\tau)$). In other words, the predicted cumulative incidence functions $\hat{F}_k(\tau)$ should approximate the oracle marginal function $F_k(\tau)$. The following can be seen as the extension of the 1-calibration \citep{haider_eective_nodate} in survival analysis to the multi-class setting. Given the presence of multiple competing events, we adopt the terminology of multiclass calibration \citep[as defined in ][]{perez-lebel_beyond_2023}. We introduce three notions of calibration in the competing-risks setting, focusing on the marginal probabilities, ordered from the weakest to the strongest one:

\begin{definition}[${\rm cal}_{K}^\alpha$-calibration(s)]\label{def:kcalib} 
    Given a time $\tau$ and an event $k$, we introduce the ${\rm cal}_{k}(\tau)$ as:
    \begin{flalign*}
        \text{\footnotesize{${\rm cal}_{k}(\tau)$-calibration:}} && {\rm cal}_{k}(\tau) = \left|F_k(\tau) - \mathbb{E}_{\mathbf{X}}(\hat{F}_k(\tau|\mathbf{X}))\right| && 
    \end{flalign*}
    With the previous definition, we define the ${\rm cal}_{k}^\alpha$-calibration, that coincides with classwise-calibration in multiclass settings \citep{zadrozny_transforming_2002}:
    \begin{flalign*}
        \text{\parbox{\linewidth}{\footnotesize{${\rm cal}_{k}^\alpha$-calibration: (Classwise-Calibration)}}}\\[-.1em]  {\rm cal}_{k}^\alpha = \left(\int_{0}^{t_{max}} {\rm cal}_{k}(\tau)^\alpha \d\tau \right)^{1/\alpha}
    \end{flalign*}
    Extending the sum over the different events, we obtain a jointly calibrated measure \citep{kull_beyond_2019, vaicenavicius_evaluating_2019}: 
    \begin{flalign*}
        \text{\footnotesize{${\rm cal}_{K}^\alpha$-calibration: (Jointly Calibration)}} && {\rm cal}_{K}^\alpha = \sum_{k=1}^K {\rm cal}_{k}^\alpha &&
    \end{flalign*}
\end{definition}
A model is perfectly ${\rm cal}_{K}^\alpha$-calibrated (resp. ${\rm cal}_{k}^\alpha$-calibrated, ${\rm cal}_{k}(\tau)$-calibrated) if ${\rm cal}_{K}^\alpha = 0$ (resp. ${\rm cal}_{k}^\alpha = 0$, \mbox{${\rm cal}_{k}(\tau)=0$}).
Assessing the calibration should be aligned with the specific research question. When the focus is on understanding the behavior of the model at a specific time point $\tau$, the ${\rm cal}_{k}(\tau)$-calibration will be most useful. To evaluate the calibration of the model for a particular, potentially rare, risk event over time, the ${\rm cal}_{k}^\alpha$-calibration offers a better perspective. For a global assessment of a model's calibration, aggregating all available information, the ${\rm cal}_{K}^\alpha$-calibration is the appropriate metric. Each of these distinct calibration measures provides a specialized viewpoint, but in the following part of this paper, we will then focus on the ${\rm cal}_{K}^\alpha$-calibration. The choice of focusing on the ${\rm cal}_{K}^\alpha$-calibration can be justified by the fact that if a model is ${\rm cal}_{K}^\alpha$-calibrated, so it will be ${\rm cal}_{k}^\alpha$-calibrated for each event of event and ${\rm cal}_{k}(\tau)$-calibrated for any fixed $\tau$.

By definition, the oracle function is calibrated according to the different distributions. \\
Any consistent marginal estimator is also asymptotically calibrated according to all definitions \ref{def:dcal}. 

\paragraph{Plug-in (PI) estimators for the ${\rm cal}_{K}^\alpha$-calibration}
From a population point of view, these calibration quantities can be estimated using any marginally consistent estimator. To remain general, we present the definitions using an abstract “plug-in” estimator, but this plug-in estimator can be implemented as the \citet{aalen_empirical_1978} estimator in practice (see Appendix~\ref{def:aj} for its formal definition).
For each event type $k$, we denote the chosen marginally consistent estimator (PI) of cumulative distribution function by~$\hat{F}^k_{\rm PI}$.

Our plug-in ${\rm cal}_{k}^\alpha$-calibration estimator simply takes such an arbitrary marginal estimator $\hat{F}^k_{\rm PI}$ (e.g., Aalen–Johansen) and plugs it into the definition of ${\rm cal}_{K}^\alpha$-calibration to obtain an empirical estimate of the calibration error. Intuitively, this is natural because, under mild assumptions, $\hat{F}^k_{\rm PI}$ is a consistent estimator of the true marginal CIFs; hence, if a model is well ${\rm cal}_{K}^\alpha$-calibrated, its predicted marginal CIFs should match $\hat{F}^k_{\rm PI}$ in the limit, which is exactly what the plug-in estimator tests.

\begin{definition}[Plug-in calibrations]\label{def:ajcal}
    The ${\rm cal}_{k}(\tau)$ calibration can be approximated by:
    \begin{flalign*}
        \text{\parbox{\linewidth}{\footnotesize{Plug-in-${\rm cal}_{k}(\tau)$ calibration:}}}\\[-.1em] {{\rm PI}_k(\tau)} = \left| \hat F^k_{{\rm PI}(\D_{cal})}(\tau) - \frac{1}{|\D_{cal}|}\sum_{\mathbf{x}_i \in \mathcal{X}_{cal}}(\hat{F}_k(\tau|\mathbf{x}_i))\right|
    \end{flalign*}%
    \begin{flalign*}
        \text{\footnotesize{Plug-in-$k^\alpha$ calibration:}} && {{\rm PI}^\alpha_k} = \left(\int_0^{t_{max}} {\rm PI}_k(\tau)^\alpha \d \tau \right)^{1/\alpha} &&
    \end{flalign*}
    Where in practice, we compute the integral by a Riemman 
    In the same vein, we can introduce the approximation of the joint calibration as: 
    \begin{flalign*}
        \text{\footnotesize{Plug-in-${\rm cal}_{K}^\alpha$-calibration:}} &&  {{\rm PI}^\alpha_{cal}} = \sum_{k=1}^K  {{\rm PI}^\alpha_k} && 
    \end{flalign*}
\end{definition}
We obtain an equivalent proposition as \ref{prop:consisD} here with the PI-calibrations:
\begin{proposition}[Consistency of the PI-${\rm cal}_{K}^\alpha$-calibration]
    As the sample size ($m$) of the calibration set $\D_{cal}$ goes to infinity, for a fixed $\alpha$ the Plug-in-${\rm cal}_{K}^\alpha$-calibration (Def. \ref{def:ajcal}) ${\rm PI}_{cal}^\alpha$ converges almost surely toward ${\rm cal}_{K}^\alpha$-calibration (Def. \ref{def:kcalib}) ${\rm cal}_{K}^\alpha$, making it a consistent estimator of the ${\rm cal}_{K}^\alpha$-calibration. 
\end{proposition}
\begin{proof}
    The proof is straightforward using the strong law of large numbers.
\end{proof}

Using the \citet{aalen_empirical_1978} estimator as the plug-in estimator, we define the \textit{AJ-$k$-calibration} for each event $k$ and the \textit{AJ-${\rm cal}_{K}^\alpha$-calibration}.

\paragraph{Plug-in-${\rm cal}_{K}^\alpha$-calibration vs CR D-calibration}
These calibrations differ in the sense that the ${\rm cal}_{K}^\alpha$-calibration (Def. \ref{def:ajcal}) compares the marginal probabilities in time and the CR D-calibration (Def. \ref{def:dcal}) assesses whether the estimator returns well-behaved cumulative distribution functions. Nonetheless, we can derive several asymptotic results that relate both calibrations and D-calibration \citep[ and Appendix \ref{def:dcalsurvival}]{haider_eective_nodate}. Appendix \ref{fig:linkcals} gives an explanatory diagram with these relationships.

\begin{restatable}[Equivalence between calibrations]{theorem}{kcaldcal}\label{thm:kcaldcal}
    If a model is ${\rm cal}_{K}^\alpha$-calibrated, then it is CR D-calibrated. \\
    Asymptotically, if a model is ${\rm cal}_{K}^\alpha$-calibrated, it implies that it is CR D-calibrated. \\
    Under the assumption that each $\hat{F}_k$ is strictly increasing and continuous, a model that is CR D-calibrated is also ${\rm cal}_{K}^\alpha$-calibrated. \\
    Additionally, if a model is CR D-calibrated, it is also calibrated according to the calibration plots. 
\end{restatable}

\begin{proof}[Proof]
    Using the idea that when a model is ${\rm cal}_{K}^\alpha$-calibrated, the estimated marginal probabilities are equal to the oracle ones and that a consistent marginal estimator is marginally calibrated (Theorem~\ref{thm:bigth}), we obtain that if a model is ${\rm cal}_{K}^\alpha$-calibrated, it implies that the model will be D-calibrated. \\
    For the other way, see proof in Appendix \ref{proof:links}.
\end{proof}

\paragraph{Choice of $\alpha$:} We use $\alpha=2$ as the standard norm, but acknowledge that a more suitable value for $\alpha$ may exist depending on the specific use-case and could be found through further experimentation.

\paragraph{Test if a model is calibrated} 
While the associated metrics give a comprehensive evaluation, a statistical test can conclude on whether observed decalibration is explained by sampling noise (limited size of calibration set).
The presented CR-D-Calibration and PI-${\rm cal}_{K}^\alpha$-calibration can be tested by applying a \citet{kolmogorov}-\citet{smirnov} (KS) test (with multiplicity adjustment) to check for the uniform distribution or to compare the \citet{aalen_empirical_1978} estimator with predictions, respectively (details in Appendix \ref{sec:tests}).

\section{HOW TO RECALIBRATE A MODEL?}
Having defined and estimated two calibration measures for competing‑risks models, we now ask: How can we recalibrate these models while preserving desirable properties?

\subsection{Formulating recalibration as a minimization problem}\label{sec:recalibration_problem}
To recalibrate a competing-risks model $(\hat{F}_k)_k$, typically trained on a dataset $\D_{\text{train}}$, on a calibration set $\D_{\text{cal}} = \{(\mathbf{x}_i, t_i, \delta_i)\}_i$, a classical post-hoc procedure consists in learning a transformation of the model outputs that minimizes a calibration loss on $\D_{\text{cal}}$.

Formally, let $\ell$ be a loss function that takes as first argument the predicted CIFs on the calibration set, $(\hat{F}_k(\cdot \mid \mathbf{x}_i))_{i,k}$, and as second argument the corresponding observed outcomes, $(t_i, \delta_i)_i$. In what follows, we take $\ell$ to be one of the calibration metrics introduced in Section~3. Let $\mathcal{G}$ be a family of recalibration maps $g$ acting on the predicted CIFs. For any $g \in \mathcal{G}$, we define the recalibrated model by
\[
\hat F^{(g)}_k(\tau \mid \mathbf{x})\overset{\rm def}{=} g\bigl(\tau, k, \hat F_k(\tau \mid \mathbf{x})\bigr),
\]
and we then choose
\[
g^* \overset{\rm def}{=} \arg\min_{g \in \mathcal{G}}~\ell\Bigl((\hat F^{(g)}_k(\cdot \mid \mathbf{x}_i))_{i,k},~(t_i,\delta_i)_i\Bigr).
\]
The final calibrated estimator is $\hat F^{\rm cal}_k \overset{\rm def}{=} \hat F^{(g^*)}_k$.

\subsection{Recalibration with asymptotic conformal intervals (AJ-recalibration)}
We introduce \emph{AJ-recalibration}: a post-hoc recalibration method using the AJ-k calibration (Def. \ref{def:ajcal}) as the function $\ell$ as detailed below. 
After training a model on $\D_{train}$, for each event $k$, we compute the AJ-$k$ calibration at different times $\tau$ on $\D_{cal}$. At calibration time, the predictions are recalibrated simply by computing the AJ-$k$ calibrations at different times so that $g$ is simply $g\big(\tau, k, \hat F_k(\tau\mid x)\big) = \hat F_k(\tau\mid x) - {\rm AJ}_{k}(\tau)$ (taking \citet{aalen_empirical_1978} as the PI estimator in Def.\ref{def:ajcal}). We give in Appendix \ref{sec:aj-recal} a full presentation of the method.   

\paragraph{Asymptotic Conformal Intervals}
To our knowledge, no formal definition of conformal prediction has been given in the competing-risks setting. For a level $\gamma \in [0,1]$ and an event $k \in \llbracket 1, K \rrbracket$, following the works of Booth \citep{booth_using_2023, candes_conformalized_2023, farina_doubly_2025, qi_conformalized_2024, qin_conformal_2024} on conformal prediction in survival analysis, we define an asymptotically marginally calibrated Upper Predictive Bound (UPB) as a function $\hat{L}$ estimated from $\D_{cal}$ satisfying
\[
\liminf_{|\D_{cal}|\to\infty} \mathbb{P}\bigl(T^* \leq \hat{L}(\mathbf{X}) \mid \Delta^* = k\bigr) \ge 1 - \gamma,
\]
where $(\mathbf{X}, T^*, \Delta^*) \sim \mathcal{D}$ is an independent observation drawn from the same distribution.
\begin{lemma}[Asymptotic Conformal intervals]
    Under Assumption \ref{ass:exchangeability}, if a model is AJ-recalibrated, it implies that it has naturally asymptotically marginally calibrated UPBs.
\end{lemma}
\begin{proof}
    By combining the consistency and the Gaussian Central limit theorem for the Aalen–Johansen estimator established in \citet{aalen_empirical_1978}, we directly derive UPBs whose width depends on the sample size of $\mathcal{D}_{cal}$.
\end{proof}

\paragraph{Discrimination is unchanged} As mentioned before, discrimination of patients is an important property in the competing-risks setting. We use the definition of the C-index designed for the competing-risks setting introduced by \citet{wolbers_prognostic_2009} and recall its formula in Appendix \ref{def:cindex}. With our recalibration method, we can prove the following:

\begin{theorem}
    The C-index in the competing-risks setting is not affected by the AJ recalibration.
\end{theorem}
\begin{proof}
    The proof is inspired by \citet{qi_conformalized_2024} in the survival setting.  The comparison between samples is performed at a fixed time point $\tau$, and we apply the same transformation at every point at $\tau$. Thus, the order is intact and the C-index is unchanged.
\end{proof} 

\subsection{Competing-risks temperature scaling}
Temperature scaling, proposed by \cite{guo_calibration_2017}, is a recalibration method adapted to multiclass classification. Due to the presence of time and censored people, we have adapted the temperature scaling to the competing-risks setting. Indeed, we propose here a slightly more complex post-hoc recalibration method using the temperature scaling framework.\\

We recalibrate independently at each prediction time $\tau$ using the general setup of Section~\ref{sec:recalibration_problem}, taking as loss function $\ell$ the plug-in ${\rm cal}_{k}(\tau)$-calibration ${\rm PI}_k(\tau)$, summed over all events $k$. As recalibration maps, we consider $\mathcal{G} = \{g_\beta\}_{\beta > 0}$, where for any vector of positive probabilities $p$,
\[
g_\beta(p) \overset{\rm def}{=} \softmax\bigl(\beta \cdot {\rm logit}(p)\bigr).
\]
In the competing-risks setting, we naturally extend this construction by applying $g_\beta$ at each time $\tau$ to the vector of CIFs values $(\hat F_k(\tau \mid \mathbf{x}))_k$, with the $\softmax$ taken over the event index $k$.

\begin{definition}[Temperature scaling]\label{def:tempsca}
    We define, for each $\tau$, temperature-scaling recalibration in the competing-risks setting as:
    \begin{align*}
        \Big(\hat F^{\rm TS-cal}_k(\tau)\Big)_k \overset{\rm def}{=} g_{\beta^*(\tau)}\Big(\hat F_\cdot(\tau)\Big)\\
        \text{with}\quad \beta^*(\tau)\overset{\rm def}{=} \argmin_{\beta>0}\sum_{k=1}^K {\rm PI}_k(\tau)\bigg(g_\beta\Big(\hat F_\cdot(\tau)\Big)\bigg).
    \end{align*}
\end{definition}

Instead of relying on a direct target (as done in the multiclass setting), our method utilizes a marginally consistent estimator of the whole population. We show in Appendix \ref{sec:temp-scaling} that, with the true probabilities or a consistent estimator of the CIFs, the probabilities remain unchanged ($\forall \tau, \beta(\tau) = 1$). \\
Finally, this design inherently ensures that the sum of probabilities for each individual equals 1.

\section{EXPERIMENTAL STUDY}
\paragraph{Code reproducibility} The code will be available online as an open-source Python library. For now, it is joint with the main paper as supplemental material.

\paragraph{Baselines}
We compare different state-of-the-art competing risks models. 
All models have been trained with their default parameters on an internal cluster (40 CPUs, 252 GB of RAM, 4 NVIDIA Tesla V100 GPUs). We used the \citet{aalen_empirical_1978} estimator. We also used a linear model, the \citet{fine_proportional_1999} estimator, which assumes proportional hazards (\emph{i.e.}, the instantaneous risk of every event) as \citet{cox_regression_1972}). We also trained a tree-based model \citep[Random Survival Forests,][]{ishwaran_random_2008} that optimizes the discrimination of patients. We also studied DeepHit \citep{lee_deephit_2018}, a neural network designed to optimize a loss function that emphasizes discrimination and negative cross-entropy. We also added SurvivalBoost \citep{alberge_survival_2025}, a tree-based model trained with a proper scoring rule. Finally, SurvTRACE \citep{wang_survtrace_2022} is a transformer model that models the observed quantiles of the events.

\paragraph{Datasets} 
The experiments were made on a synthetic dataset and two real-life datasets (see Appendix \ref{sec:datasetsdescr}):
\begin{itemize}[leftmargin=1.6ex,topsep=.2ex]
    \item Synthetic: We design a synthetic dataset with linear features with dependent censoring, and use it in the following with three competing events. 
    \item \cite{seer_dataset}: This real-life dataset contains 470,000 patients with a 10-year follow-up. After preprocessing, it yields three competing events: death from breast cancer, cardiac events, and other causes.
    \item METABRIC \citep{rueda2019dynamics}: With 1,700 patients, this real-life dataset contains two competing events, with a follow-up period of up to 30 years. 
\end{itemize}

\paragraph{$\hat{D}$-calibration and AJ-calibration on real-life datasets}
\autoref{tab:calibseer10k} shows the CR $\hat{D}$-calibration and the AJ-$k$-calibration for for the SEER dataset trained over 10k samples. We also show whether the KS test assessed the calibration of the model. In this context, we see that the model that is the most calibrated overall is the \citet{aalen_empirical_1978} estimator, and that DeepHit and SurvTRACE are the models that are less calibrated. Other results (METABRIC, SEER 100k training datapoints, and the synthetic dataset) can be found in Appendix \ref{sec:resultsappendix}. The method least CR $\hat{D}$-calibrated, SurvTRACE, is also the farthest from the \citet{aalen_empirical_1978} in Fig. \ref{fig:meanincfunc}. Across all datasets (Tables \ref{tab:metabriccalibperevent} through \ref{tab:seer100kcalibperevent}), the computation of the CR-$\hat{D}$-calibration for each event reveals a consistent trend: most models, excepted DeepHit, exhibit their highest calibration error for the minority class. This is a concerning issue in the competing risks setting.

\begin{table}[t!]
    
    \centering
     \caption{\textbf{Competing risks calibrations:} CR-$\hat{D}_{cal}$ and AJ-${\rm cal}_{K}^\alpha$-calibrations computed on the SEER dataset with 10k training samples with $\alpha = 2$ over 5 random seeds. We recover that \citet{aalen_empirical_1978} is the most calibrated model while SurvTRACE has the worst metrics (expected given Figure \ref{fig:meanincfunc}). The table also reports the proportion of seeds passing the 5\% calibration test (Appendix~\ref{sec:tests}), with higher values indicating better calibration.
     }
    \label{tab:calibseer10k}%
    \small
    \begin{tabular}{@{~}l@{~~~}r@{~}r@{~~~}r@{~}r@{~}}
    \toprule
    & \multicolumn{2}{c}{AJ-K-cal} & \multicolumn{2}{c}{CR-$\hat{D}_{cal}$} \\
    \cmidrule(r){2-3}\cmidrule(l){4-5}
    Model & mean $\pm$ std & test & mean $\pm$ std & test \\
    \midrule
    AalenJohansen & 0.00 $\pm$ 0.00 & 100\%
                   & 0.05 $\pm$ 0.01 & 100\% \\
    DeepHit       & 0.18 $\pm$ 0.13 & 0\%
                   & 1.13 $\pm$ 0.56 & 0\%\\
    FineGray      & 0.00 $\pm$ 0.00 & 100\%
                   & 1.18 $\pm$ 2.30 & 20\% \\
    RSF           & 0.00 $\pm$ 0.00 & 40\%
                   & 0.09 $\pm$ 0.04 & 80\% \\
    SurvTRACE     & 0.24 $\pm$ 0.17 & 0\%
                   & 1.34 $\pm$ 0.06 & 0\%\\
    SurvivalBoost & 0.01 $\pm$ 0.00 & 0\%
                   & 0.09 $\pm$ 0.01 & 100\% \\
    \bottomrule
    \end{tabular}
    
\end{table}

\paragraph{Recalibrations: impact on the metrics}
We study the probabilities before and after the recalibration. For this, we use the Integrated Brier score \citep[IBS, ][]{graf_assessment_1999, schoop_quantifying_2011}. This is the sum over time of the Brier Score re-weighted to account for the censoring. It is a strictly proper scoring rule: its minimum gives the true probabilities. Recalibration might influence the IBS as it alters the probabilities for a poorly calibrated model. Figure \ref{fig:metabricrecal} shows the metrics on the SEER dataset (Appendix \ref{sec:moreresults} gives the results on the METABRIC dataset, the synthetic dataset, and for different training sizes of the SEER dataset). Figure \ref{fig:metabricibs} shows the IBS (lower is better) before and after recalibration over 5 different seeds. It shows that, after recalibration, probability errors for the different models are smaller, while having the AJ-${\rm cal}_{K}^\alpha$-calibration and the D-CR-calibration smaller (Fig. \ref{fig:metabricajcal} \ref{fig:metabricdcr}). \\
As mentioned before, both recalibrations, applied to predicted probabilities, leave discrimination power (\emph{i.e.} C-index) unchanged. Nevertheless, if the corresponding input probabilities do not initially sum to 1 across predicted classes for each individual, their renormalization does change discrimination
(Figure \ref{fig:metrics_metabric}). 

\paragraph{Comparison between recalibration methods}
The best recalibration method depends on the dataset size and desired performance. Across our experiments (Section \ref{sec:moreresults}), AJ-recalibration yielded superior calibration metrics universally and lower IBS on smaller datasets (Fig. \ref{fig:metabricibs} and Fig. \ref{fig:metrics_synthe}). Conversely, the Temperature Scaling approach seems more effective for larger datasets (Fig. \ref{fig:metrics_seer100k}) and has the added benefit of producing meaningful individual probabilities.

\begin{figure}[t!]
    \begin{subfigure}{0.50\textwidth}
        \includegraphics[width=\textwidth]{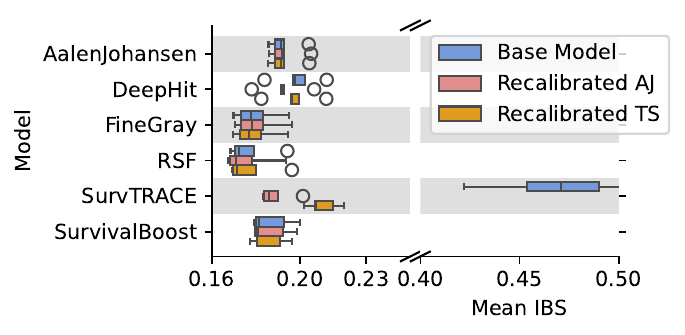}\vspace*{-4ex} 
        \caption{\sffamily Integrated Brier Score\hspace*{.55\linewidth}}\label{fig:metabricibs}%
        \vspace*{.5ex} 
    \end{subfigure}
    \begin{subfigure}{0.40\textwidth}
        \includegraphics[width=1.2\textwidth]{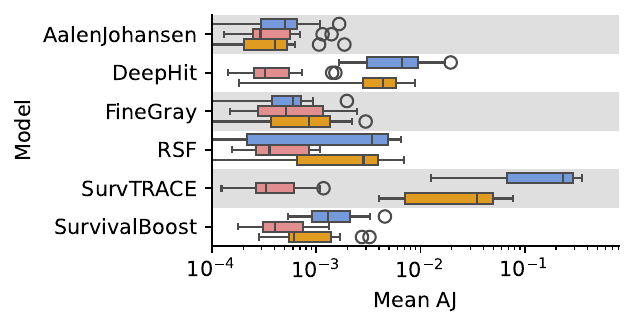}\vspace*{-4.3ex}
        \caption{\sffamily AJ-${\rm cal}_{K}^\alpha$-calibration\hspace*{.53\linewidth}}\label{fig:metabricajcal}%
        \vspace*{.5ex} 
    \end{subfigure}
    \begin{subfigure}{0.50\textwidth}
        \includegraphics[width=\textwidth]{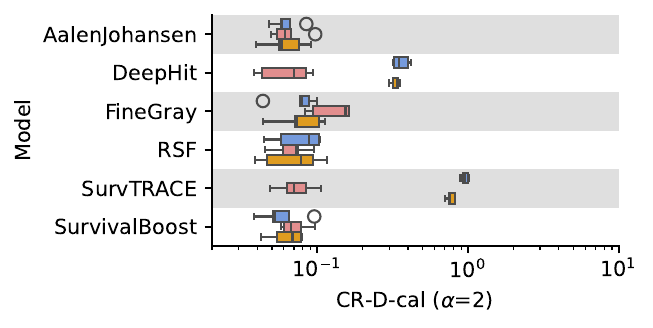}\vspace*{-4.3ex}
        \caption{\sffamily CR-D-calibration\hspace*{.59\linewidth}}\label{fig:metabricdcr}%
        \vspace*{.5ex} 
    \end{subfigure}
    \caption{\textbf{IBS and calibration on METABRIC.}
    Boxplots over 5 seeds of the integrated Brier score (IBS), AJ-${\rm cal}_{K}^\alpha$-calibration (with parameter $\alpha = 2$), and CR-D-calibration ($\alpha = 2$) for each model and its AJ-recalibration and temperature-scaling (TS) recalibration.
    Results on the synthetic and SEER datasets are reported in Appendix~\ref{sec:moreresults}.}\label{fig:metabricrecal}
\end{figure}

\section{DISCUSSION AND CONCLUSION}

\paragraph{Limitations} We have introduced 2 metrics, with complementary strengths and weaknesses.
A limitation of CR-D-calibration, as in \citet{haider_eective_nodate}, is that we assume a prediction horizon $t_{\max}$ large enough for cumulative incidence functions to approach the probability of the event of interest (\emph{i.e.} $\forall x, \hat{F}_k(t_{max}|\mathbf{x}) \approx \hat{F}_k(\infty|\mathbf{x})$), which may not hold in practice.
In contrast, the PI-${\rm cal}_{K}^\alpha$-calibration does not have this limitation by not relying on the behavior of the cumulative incidence function at $t_{\max}$. However, it requires sufficiently large data to ensure that the PI estimator closely approximates the true marginal distribution.

A natural direction for future work is a dedicated simulation study under known ground truth that systematically investigates the finite-sample behavior of our calibration estimators and recalibration procedures—varying sample sizes, as well as covariate, event-time, and censoring distributions—and assesses whether pathologies similar to those reported for classification calibration errors by \citet{gruber_Buettner_22} arise in our setting, in which case it would be important to explore more robust metrics and tests for competing risks.

\paragraph{Acknowledgments}

JA, JA, and GV acknowledge funding from the European Union’s Horizon Europe research and innovation program under grant agreement No 1010954433 (Intercept-T2D). TH acknowledge funding from the French Ministry of Health.

\paragraph{Conclusion} 
We addressed a critical gap in survival analysis with two novel, oracle-minimized calibration metrics for competing-risks settings: the \textit{CR-D-calibration} and the \textit{${\rm cal}_{K}^\alpha$-calibration}, proving their equivalence under a mild assumption. For practical use, we provided for both consistent estimators and corresponding recalibration formulation and procedures (one derived from AJ-$k$-calibration and an adaptation of the temperature scaling). We have shown that those methods improve model calibration while preserving discriminative capacity and increasing their abilities to predict individualized probabilities. \\
Our application of these new metrics to SOTA competing risks models on both synthetic and real datasets revealed heterogeneity in their calibration, increasing the urgent necessity of this evaluation criteria. Ultimately, relying on poorly calibrated predictions risks inadequate decisions in critical applications, emphasizing the need to consider our proposed calibration measures alongside other performance criteria.

\bibliography{conformelle}
\bibliographystyle{apalike}




    Yes, all proofs can be found in the appendix, and sketch proofs in the main manuscript. 
\addtocontents{toc}{\protect\setcounter{tocdepth}{2}}



\appendix
\onecolumn
\section*{SUPPLEMENTARY MATERIALS}

\section{FORMAL DEFINITIONS}
\begin{definition}[D-calibration: Survival setting]\label{def:dcalsurvival}
    For a dataset $D = \{ [\mathbf{x}_i, t_i, \delta_i] | i = 1, ..., n\}$, and any interval $[a, b] \subset [0,1]$: 
    \begin{equation}
        D([a,b]) = \{ [\mathbf{x}_i, t_i, \delta_i = 1] \in D \quad | \quad \hat{S}(t_i|\mathbf{x}_i) \in [a, b] \} 
    \end{equation}
\end{definition}

\begin{definition}[C-index CR \cite{wolbers_prognostic_2009}]\label{def:cindex}
    We recall the definition of the C-index at time $\tau$ for the $k^{th}$ competing event \cite{wolbers_prognostic_2009} as:
\begin{equation}
    \mathrm{C}(\tau) = \frac{\sum_{i=1}^n \sum_{j=1}^n (A_{ij}
        \hat{W}_{ij, 1}^{-1} + B_{ij} \hat{W}_{ij, 2}^{-1}) Q_{ij}(\tau)
        \mathbb{1}_{t_i \leq \tau, \delta_i = k}}
        {\sum_{i=1}^n \sum_{j=1}^n (A_{ij}
        \hat{W}_{ij, 1}^{-1} + B_{ij} \hat{W}_{ij, 2}^{-1}) \mathbb{1}_{t_i \leq \tau, \delta_i = k}}
\end{equation}
where:
\begin{align}
        A_{ij} &= \mathbb{1}_{t_i < t_j \cup (t_i = t_j \cap \delta_j = 0)} \\
        B_{ij} &= \mathbb{1}_{t_i \geq t_j, \delta_j \neq k, \delta_j \neq 0} \\
        \hat{W}_{ij,1} &= \hat{G}(t_i| \mathbf{X} = \mathbf{x}_i) \hat{G}(t_i|\mathbf{X} = \mathbf{x}_j) \\
        \hat{W}_{ij,2} &= \hat{G}(t_i|\mathbf{X} = \mathbf{x}_i) \hat{G}(t_j|\mathbf{X} =\mathbf{x}_j) \\
        Q_{ij}(t) &= \mathbb{1}_{F_k(\tau| \mathbf{X} = \mathbf{x}_i) > F_k(\tau| \mathbf{X} = \mathbf{x}_j)}
\end{align}
where $\hat{G}$ is one estimator of the censoring function, $G(\tau | \mathbf{X}= \mathbf{x}) = \mathbb{P}(C > \tau | \mathbf{X}= \mathbf{x})$.
\end{definition} 

\subsection{Marginal estimators in survival analysis and in the competing risks setting}
Here is the definition of the \citet{kaplan_nonparametric_1958} estimator, a marginal estimator in the survival analysis setting, 
\begin{definition}[\citet{kaplan_nonparametric_1958} estimator]\label{def:km}
\begin{align}
    \textbf{Kaplan-Meier (KM) estimator:} && \hat F_{\rm KM}(\tau) = 1 - \hat{S}_{KM}(\tau) = 1 - \prod\limits_{t_i \le \tau} \frac{Y(t_i)-d(t_i)}{Y(t_i)}&&
\end{align}
$d(t)$ is the number of events at time $t_i$ and $Y(t_i)$ is the number of individuals at risk at time $t_i$.
\end{definition}

And now, in the competing risks setting, where there exist more than one event of interest, we write the definition of the \citet{aalen_empirical_1978} estimator.
\begin{definition}[\citet{aalen_empirical_1978} estimator]\label{def:aj}
\begin{align}
    \textbf{Aalen-Johansen (AJ) estimator:} && \hat F^k_{\rm AJ}(\tau) = \sum_{t_i \leq \tau} \hat{S}(t_i^-) \frac{d_k(t_i)}{Y(t_i)}&&
\end{align}
where $\hat{S}(t_i^-)$ is the survival function estimated from the \citet{kaplan_nonparametric_1958} (see  \ref{def:km}) estimator, $d_k(t)$ is the number of events of type $k$ at time $t_i$ and $Y(t_i)$ is the number of individuals at risk at time $t_i$.
\end{definition}

\section{SURVIVAL ANALYSIS RECALIBRATION METHOD} 
Here is the result of the recalibration using \citet{qi_conformalized_2024}'s method. Using the D-calibration in survival analysis, onto the SEER dataset with 10k training samples and 3 competing risks, we clearly see that the method has a negative impact on the marginal calibration of the \citet{aalen_empirical_1978} estimator. 

\begin{figure}[h!]
    \centering
     \includegraphics[width=.3\textwidth]{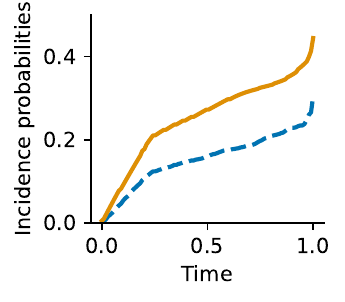}%
	\llap{\raisebox{.26\linewidth}{Event 1}\qquad}%
     \includegraphics[width=.3\textwidth]{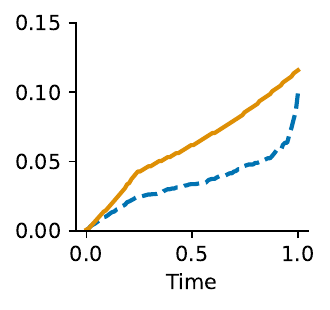}%
 	\llap{\raisebox{.26\linewidth}{Event 2}\qquad}%
     \includegraphics[width=.3\textwidth]{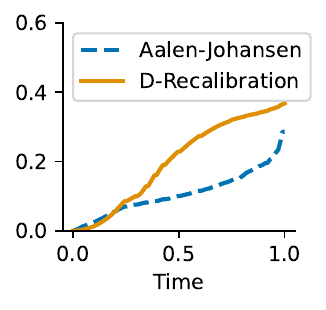}%
 	\llap{\raisebox{.26\linewidth}{Event 3}\qquad}%
     \caption{\textbf{Survival-analysis recalibration of good competing risk probabilities breaks them}. On this real-life dataset (SEER, 10k training samples), we apply the D-recalibration method in survival analysis \cite{qi_conformalized_2024} to \citet{aalen_empirical_1978} by treating the risks independently (cause-specific). It incorrectly changes the marginal probabilities, though these were correct, as we used a consistent marginal estimator.
 The root of the problem is that survival-risk recalibration is applied to events independently, while their probabilities are tied \cite{gorfine_frailty-based_2011}. \label{fig:recalsurvival}}
\end{figure}

\section{DISTRIBUTION CALIBRATION IN THE COMPETING-RISKS SETTING (CR D-Cal)} \label{dcalsurvival}
In this part, we reuse the notations that were explained in the main paper.
\subsection{Without censoring}
For now, we suppose there are no censored individuals. It is crucial to understand the hurdles in extending the D-calibration into the competing-risks setting. As mentioned earlier, in the survival analysis ~\ref{related_work}, the authors assume that the event of interest will almost surely happen, \emph{i.e.} $F(\infty |\mathbf{x}) = 1$, they had that $F(T^*)\sim \mathcal{U}(0, 1)$. Whereas, in the competing setting, that for each event $k$, $F_k(T^*, \mathbf{X}) \nsim \mathcal{U}(0, 1)$ because $F_k(T^*, \mathbf{X})\leq F_k(\infty |\mathbf{X}) = \mathbb{P}(\Delta^* = k | \mathbf{X} = \mathbf{x}) < 1$. \\
To overcome this issue, we introduce the D-calibration without any censored individuals as \ref{def:dcalnc}:

\begin{definition}[CR D-calibration - No Censoring]\label{def:dcalnc}
    Thus, in the case of non-censorship, we define the Competing Risks Distribution calibration (CR D-calibration NC) as:  
    \begin{align}
        \text{CR D-calibration NC} &&  D^{NC}_{cal} = \sum_{1}^K \left\lVert\frac{\mathbb{E}_{T^*, \Delta^*}(B^k_{[0, \rho]})}{\mathbb{E}_\mathbf{X}(\hat{F}(\infty |\mathbf{X}))} - \rho\right\rVert_{\alpha} && \\
        \text{where } && B^k_{[0, \rho]} = \mathbb{1}_{\hat{F}_{k/\infty}(T^*|\mathbf{X}) \in [0, \rho]}\mathbb{1}_{\Delta^*=k}
    \end{align}
     If the $k^{th}$ incidence function is D-calibrated, we have $B^k_{[0, \rho]} \simeq \rho$. \\ 
     This distance corresponds to integrating the differences between the uniform ratio and the cumulative incidence function of a uniform law. 
\end{definition}

The CR D-calibration only have a meaning if the oracle functions (\emph{i.e.} each of the CIF) are calibrated according to our definition. The following lemma \ref{lemm:ratiocdf} helps us to overcome this issue:

\begin{restatable}[Uniform Ratio CDF]{lemma}{lemmuniformratio}\label{lemm:ratiocdf}
    For the oracle function, conditioning of $\Delta^*$ and $\mathbf{X}$, this ratio is following a uniform distribution, \emph{i.e.}: 
    \begin{align}
       F_{k/\infty}(T^*|\mathbf{X}) ~\big|~ \Delta^*=k, \mathbf{X}  \quad \sim \quad \mathcal{U}(0, 1)
    \end{align}
\end{restatable}

\begin{proof}[Proof Sketch]
    Using the Uniform CDF theorem \citet{wikith} and the conditioning over $\Delta^*$ and $X$. The whole proof can be found in Appendix. 
\end{proof}

\begin{restatable}[CR D-calibration - No Censoring]{theorem}{thdcnc}\label{thm:dcnc}
    With lemma \ref{lemm:ratiocdf}, the oracle functions are CR D-calibrated. 
    This proves that the result that $D_{NC}^{CR}$ is proper. \\
    A marginally consistent estimator is asymptotically NC CR D-calibrated. 
\end{restatable}

\begin{proof}[Proof sketch]
    For the oracle functions, the result is straightforward with the Lemma \ref{lemm:ratiocdf}. 
    With Lebesgue's dominated convergence theorem, we show that a marginally consistent estimator is asymptotically CR D-calibrated.  For more details, see below.
\end{proof}

\begin{corollary}
    Because the Aalen-Johansen estimator is a consistent (see Lemma \ref{lemm:ajconsistent}) marginal estimator of the oracle functions, it is also asymptotically NC CR D-calibrated.
\end{corollary}

For a given population and a given model, we also have to define the estimation of the D-calibration according to the population (named $\hat{D}$-calibration). 

\begin{definition}[CR $\hat{D}$-calibration - No Censoring]
    Thus, in the case of non-censorship, for a given population, we define: 
    \begin{align}
         && b^k_{[0, \rho]} = \frac{{\rm card}\{i\mid\delta_i=k~~{\rm and}~~ \hat F_{k/\infty}(t_i\mid \mathbf{x}_i)\in [0, \rho]\}}{{\rm sum}\{\hat F_k(\infty\mid \mathbf{x}_i)\}}
    \end{align}
    Then, we can define the estimation of the D-calibration as: 
     \begin{flalign}
        \hat{D}-\text{calibration estimation} && \hat{D} \defeq \sum_{k=1}^K \, \left\lVert b^k_{[0, \rho]}  - \rho \right\rVert_{\alpha} && 
\end{flalign}
\end{definition}

\section{PROOFS}
For the following proofs, we will use the definition of the \citet{aalen_empirical_1978} estimator (written in Appendix \ref{def:aj}) given an event of interest. We first recall that, for each event of interest $k$, the \citet{aalen_empirical_1978} estimator is consistent. 
\begin{lemma}\label{lemm:ajconsistent}
    For each event of interest $k \in \llbracket 1, K \rrbracket$, $\hat{F}^k_{\AJ}$ is consistent. 
\end{lemma}

\begin{proof}
    We fix an event of interest $k$. We consider the training dataset $\mathcal{D}_{train}$ with $n$ samples. In this proof, we write the \citet{aalen_empirical_1978} estimator for an event of interest as $F_n^{k}$ (resp. $S_n$). \\
    Given the definition of \citet{aalen_empirical_1978}, the proof can be detailed in 3 main points. 
    \begin{enumerate}
        \item \textbf{Consistency of the survival function:} \citet{kaplan_nonparametric_1958} estimator is consistent \cite{wang_km_consi_1987} in the sense that the survival function converges in probability for the uniform norm toward the oracle function: for all $M$ such that $\P(T^*\geq M,~\Delta^*=k)>0$, $$\forall\varepsilon>0, \lim_{n\to \infty} \P(\sup_{\tau\leq M}|S_n(\tau) - S(\tau)| \leq\varepsilon) = 1.$$
        \item \textbf{Convergence of the instant risk rate.} With the law of large numbers, we have: $\frac{d_k(t_i)}{Y(t_i)}$ that converges almost surely toward the instant risk of contracting the event of interest, \emph{i.e.}, almost surely $$\frac{d_k(t_i)}{Y(t_i)} \longrightarrow \lambda_k(t_i)dt = dt \,\lim_{dt \longrightarrow 0} \frac{\P(t_i \leq T^*\leq t_i + dt, \Delta=k) | T^* \geq t_i)}{dt}$$
        \item \textbf{Convergence of the sum to the integral.} Finally, because we have an integrable function (Riemman), the sum converges toward the integral.  
    \end{enumerate}
    With those steps, we obtain the convergence in probabilities: $$\hat F^k_{\rm AJ}(\tau) = \sum_{t_i \leq \tau} \hat{S}(t_i^-) \frac{d_k(t_i)}{Y(t_i)} \longrightarrow \int_{0}^\tau  \hat{S}(t) \lambda_k(t)dt = F^k(\tau)$$
    Finally, we have obtain that $$\lim_{n\to\infty} \P(|F_n^k(\tau) - F_k(\tau)| = 0) = 1 .$$ The Aalen-Johansen estimator is marginally consistent. 
\end{proof}

\subsection{CR D-calibration without censoring mechanism}
\lemmuniformratio*

\begin{proof}[Proof: distribution of $F_k/F_k^\infty$]\label{proof:lemmaratio}
    We prove that $F_{k/\infty}(T^*\mid \mathbf{X}) | \Delta^* =k , \mathbf{X} \sim \mathcal{U}(0, 1)$. 
    \begin{align}
        \forall t \in [0, 1]&, \\
        \P(\frac{F_k(T^*|\mathbf{X})}{F_k(\infty| \mathbf{X})} \leq t|\Delta^* = k, \mathbf{X}) &= \P(T^* \leq F_k^{-1} (t ~ F_k(\infty| \mathbf{X})|\mathbf{X})|\Delta^* = k, \mathbf{X}) \\
        &= \P(T^* \leq F_k^{-1} (t ~ F_k(\infty| \mathbf{X})|\mathbf{X})|\Delta^* = k, \mathbf{X}) \\
        &= \frac{\P(T^* \leq F_k^{-1} (t ~ F_k(\infty| \mathbf{X})|\mathbf{X}), \Delta^* = k| \mathbf{X})}{ \P(\Delta^*=k|\mathbf{X})}\\
        &= \frac{F_k(F_k^{-1}(t ~ F_k(\infty| \mathbf{X}) |\mathbf{X}))}{F_k(\infty| \mathbf{X})} \\
        &= t
    \end{align}

    Thus, this result is constant for all $\mathbf{X}$.
\end{proof}

\thdcnc*

\begin{proof}
    The proof is a special case of the setting with censoring. See Appendix \ref{proof:thproperness} for the whole proof. 
\end{proof}

\subsection{CR D-calibration with the censoring mechanism}

Let the bucket $B_{[a,b]}$ be an extension of $B_{[0,\rho]}$ defined by:
\[
B_{[a,b]} \overset{\rm def}{=} B_{[0,b]} - B_{[0,a]}.
\]
\begin{restatable}[Expectation with the oracle function]{lemma}{lemmexpect}\label{lemm:exporacle}
For the oracle function, we prove: $$\E(B_{[a,b]}\mid \mathbf{X}) = (b-a)F_k(\infty\mid \mathbf{X}) . $$
\end{restatable}

\begin{proof}\label{proof:expec}
Starting from 
\begin{align}
    {\mathbb P}(T^*\leq t\mid\Delta^*=k,~\mathbf{X})&=\frac{F_k(t\mid \mathbf{X})}{F_k(\infty\mid \mathbf{X})} \\
    {\mathbb P}\left(\frac{F_k(T^*\mid \mathbf{X})}{F_k(\infty\mid \mathbf{X})}\in[a,b]\mid\Delta^*=k,~\mathbf{X}\right) &=b-a \quad \text{Assumption \ref{ass:continous}, Lemma \ref{lemm:ratiocdf}} \\
    \text{And} \quad {\mathbb P}\left(\frac{F_k(T^*\mid \mathbf{X})}{F_k(\infty\mid \mathbf{X})}\in[a,b],~\Delta^*=k\mid \mathbf{X}\right)& =(b-a)F_k(\infty\mid \mathbf{X}) \quad \text{Bayes' Theorem}
\end{align}

We can't work with (28) because we don't have access to ${\mathcal L}(T, \Delta\mid \Delta^*=k)$ because of the censored individuals, so we'll work with (29).

First, we study the left part of the equality, we will study several cases: $T^*<C$ (we have access to the expression) and $C\leq T^*$ (we will estimate the expression). To make the computations more readable, we will also separate the cases: when $\frac{F_k(C\mid \mathbf{X})}{F_k(\infty\mid \mathbf{X})}<a$ and $\frac{F_k(C\mid \mathbf{X})}{F_k(\infty\mid \mathbf{X})}\in [a,b]$, because $\frac{F_k(C\mid \mathbf{X})}{F_k(\infty\mid \mathbf{X})}>b$ is not possible (as $\frac{F_k(T^*\mid \mathbf{X})}{F_k(\infty\mid \mathbf{X})}\in[a,b]$) \emph{i.e.}, we can decompose the computation as:
\begin{multline}
    {\mathbb P}(\frac{F_k(T^*\mid \mathbf{X})}{F_k(\infty\mid \mathbf{X})}\in[a,b],~\Delta^*=k\mid \mathbf{X})
= \underbrace{{\mathbb P}(\cdots,~T^*<C\mid \mathbf{X})}_{\displaystyle P_1} \\ 
+ \underbrace{{\mathbb P}(\cdots,~C\leq T^*, \frac{F_k(C\mid \mathbf{X})}{F_k(\infty\mid \mathbf{X})}\in[a,b]\mid \mathbf{X})}_{\displaystyle P_2} \\ 
+ \underbrace{{\mathbb P}(\cdots,~C\leq T^*, \frac{F_k(C\mid \mathbf{X})}{F_k(\infty\mid \mathbf{X})}<a\mid \mathbf{X})}_{\displaystyle P_3}
\end{multline}

Now, $P_1$ is just:
$$
P_1 = {\mathbb P}(\frac{F_k(T\mid \mathbf{X})}{F_k(\infty\mid \mathbf{X})}\in[a,b],~\Delta=k\mid \mathbf{X})
$$

For the other parts $P_2$ and $P_3$,  we write the shifted distribution $T^*\mid \Delta^*=k, X$ :
$$
\mathbb{P}[T^*\leq t,~\Delta^*=k\mid T^*> s,~\mathbf{X}] = \frac{F_k(t\mid \mathbf{X}) - F_k(s\mid \mathbf{X})}{S^*(s\mid \mathbf{X})}
$$
And: 
\begin{align}
    \P(F_k(T^*\mid \mathbf{X})\leq u,~\Delta^*=k\mid T^*> s,~\mathbf{X}) &= \P(T^*\leq F_k^{-1}(u|\mathbf{X}), ~\Delta^*=k \mid T^*> s,~\mathbf{X})  \quad \text{\ref{ass:continous}} \\
    &= \P(s \leq T^*\leq F_k^{-1}(u|\mathbf{X}), ~\Delta^*=k \mid~\mathbf{X}) \frac{1}{S(s|\mathbf{X})} \\
    &= \frac{u - F_k(s\mid \mathbf{X})}{S(s\mid \mathbf{X})}\quad {\rm if}\quad F_k(s\mid \mathbf{X})\leq u\leq F_k(\infty\mid \mathbf{X})
\end{align}

We rewrite $P_2$ conditioning by $C$ and $C\leq T^*$ and we use Assumption \ref{ass:noninfocensoring} \emph{i.e.} $T^*\perp\!\!\!\perp C\mid \mathbf{X}$ :
\begin{align}
P_2 &= {\mathbb P}\left(\frac{F_k(T^*\mid \mathbf{X})}{F_k(\infty\mid \mathbf{X})}\leq b,~\Delta^*=k, C\leq T^*, \frac{F_k(C\mid \mathbf{X})}{F_k(\infty\mid \mathbf{X})}\in[a,b]\mid \mathbf{X}\right)\\
&= {\mathbb E}\left({\mathbb P}(\frac{F_k(T^*\mid \mathbf{X})}{F_k(\infty\mid \mathbf{X})}\leq b,\Delta^*=k\mid C, C\leq T^*, \mathbf{X}){\mathbb 1}(C\leq T^*, \frac{F_k(C\mid \mathbf{X})}{F_k(\infty\mid \mathbf{X})}\in[a,b])\mid \mathbf{X}\right)\\
&= {\mathbb E}\left(\frac{F_k(\infty\mid \mathbf{X})b - F_k(C\mid \mathbf{X})}{S^*(C\mid \mathbf{X})}{\mathbb 1}(C\leq T^*, \frac{F_k(C\mid \mathbf{X})}{F_k(\infty\mid \mathbf{X})}\in[a,b])\mid \mathbf{X}\right)
\end{align}

Idem, with $P_3$, we obtain that:
$$
P_3 = {\mathbb E}(\frac{F_k(\infty\mid \mathbf{X})(b - a)}{S^*(C\mid \mathbf{X})}{\mathbb 1}(C\leq T^*, \frac{F_k(C\mid \mathbf{X})}{F_k(\infty\mid \mathbf{X})}<a)\mid \mathbf{X})
$$

So, when we recall the definition of $B_{[a,b]}$: 
\begin{multline}\label{eq:bucket_ab}
    B_{[a,b]}\overset{\rm def}{=}{\mathbb 1}(\frac{F_k(T\mid \mathbf{X})}{F_k(\infty\mid \mathbf{X})}\in[a,b],~\Delta=k) \\
+ \frac{F_k(\infty\mid \mathbf{X})b - F_k(C\mid \mathbf{X})}{S^*(C\mid \mathbf{X})}{\mathbb 1}(C\leq T^*, \frac{F_k(C\mid \mathbf{X})}{F_k(\infty\mid \mathbf{X})}\in[a,b]) \\
+ \frac{F_k(\infty\mid \mathbf{X})(b - a)}{S^*(C\mid \mathbf{X})}{\mathbb 1}(C\leq T^*, \frac{F_k(C\mid \mathbf{X})}{F_k(\infty\mid \mathbf{X})}<a)\in\sigma(T, \Delta, \mathbf{X})
\end{multline}

We do obtain that:
$$
{\mathbb E}(B_{[a,b]}\mid \mathbf{X}) = (b-a)F_k(\infty\mid \mathbf{X})
$$
\end{proof}

\bigth* 

\begin{proof}\label{proof:thproperness}
    For the oracle functions, we have Lemma \ref{lemm:exporacle}: ${\mathbb E}(B_{[0, \rho]}\mid \mathbf{X}) = \rho F_k(\infty\mid \mathbf{X}) \Longrightarrow \E_\mathbf{X}({\mathbb E}(B_{[0, \rho]}\mid \mathbf{X})) = \rho \E(F_k(\infty)|\mathbf{X})$. So: 
    \begin{align}
        D^{CR} &= \sum_{1}^K \left\lVert\frac{\mathbb{E}_{T, \Delta}(B^k_{[0, \rho]})}{\E_\mathbf{X}(\hat{F}_k(\infty |\mathbf{X}))} - \rho\right\rVert_{\alpha} \\ 
        &= \sum_{1}^K \left\lVert\frac{\rho \E_\mathbf{X}(F_k(\infty|\mathbf{X}))}{\E_\mathbf{X}(F_k(\infty|\mathbf{X}))} - \rho\right\rVert_{\alpha} \\
        &= 0
    \end{align}

    Let's be $\hat{F}^n_k(\tau)$ a marginally consistent estimator of the cumulative incidence functions trained of $n$ samples. \\
    Let's suppose the bigger assumption that $C \indep T^*$.
    Because we have $\forall k, \forall \tau, \hat{F}^n_k(\tau) \to F_k(\tau)$, this convergence is point-wise. \\
    Moreover, we have $\forall k, \forall \tau, \mid \hat{F}^n_k(\tau) \mid \leq 1 $ by definition of $\hat{F}^n_k(\tau)$. \\
    So, by Lebesgue's dominated convergence theorem, we obtain that: 
    $$\forall k, \forall \tau, \quad \lim_{n\to\infty} \E(| \hat{F}^n_k(\tau) - F_k(\tau) | ) = 0$$
    With that, we obtain exactly the same computations as before with the marginal estimator. 
\end{proof}

\subsection{Links between calibration measures}
\kcaldcal*
\begin{proof}\label{proof:links}
    We will do the proof of the second part in the asymptotic case ; the non-asymptotic case is similar but simpler. \\
    Consider the sequence of estimators indexed by $n\to\infty$, which is asymptotically CR D-calibrated. The asymptotic CR D-calibrated refers to the property that 
    $$\forall \rho \in [0, 1], \E( B^{k,n}_{[0, \rho]})/\hat{F}^n_k(\infty) \longrightarrow_{n\to\infty} \rho .$$
   A straightforward computation following the same ideas as in Lemma \ref{lemm:exporacle} gives us
   \[
   \E( B^{k,n}_{[0, \rho]}) = \P(\hat{F}^n_{k/\infty}(T^*)\leq \rho,\Delta^*=k),
   \]
   so we obtain
   \begin{align}\label{eq:proof:thproperness2_1}
       \forall \rho \in [0, 1], \P(\hat{F}^n_{k/\infty}(T^*)\leq \rho,\Delta^*=k)/\hat{F}^n_k(\infty) \longrightarrow_{n\to\infty} \rho .
   \end{align}
   Applying \eqref{eq:proof:thproperness2_1} to $\rho=1$ yields $\P(\Delta^*=k)/\hat{F}^n_k(\infty) \to 1$ and so $\hat{F}^n_k(\infty)\to\P(\Delta^*=k)$ as $n\to\infty$.
   Hence, substituting the limit of the denominator into \eqref{eq:proof:thproperness2_1},
   \begin{align}\label{eq:proof:thproperness2_2}
       \forall \rho \in [0, 1], \P(\hat{F}^n_{k/\infty}(T^*)\leq \rho\mid\Delta^*=k) \longrightarrow_{n\to\infty} \rho .
   \end{align}
   Since the CIF of $T^*\mid\Delta^*=k$ is $F_{k/\infty}$ and the function $\hat{F}^n_{k/\infty}$ is stricly increasing, the equation \eqref{eq:proof:thproperness2_2} is equivalent to
   \begin{align}\label{eq:proof:thproperness2_3}
       \forall \rho \in [0, 1], F_{k/\infty}\circ(\hat{F}^n_{k/\infty})^{-1}(\rho) \longrightarrow_{n\to\infty} \rho .
   \end{align}
   Using the continuity of \(F_{k/\infty}^{-1}\), we have that \((\hat{F}^n_{k/\infty})^{-1}\) converges pointwise to \(F_{k/\infty}^{-1}\). By the classical Dini's theorem on the uniform convergence of monotone functions, this convergence is in fact uniform. Therefore, we can invert the functions and obtain the convergence of \(\hat{F}^n_{k/\infty}\) toward \(F_{k/\infty}\). This concludes the proof.
   
\end{proof}
\begin{figure}[h!]
    \centering
    \includegraphics[width=0.4\linewidth]{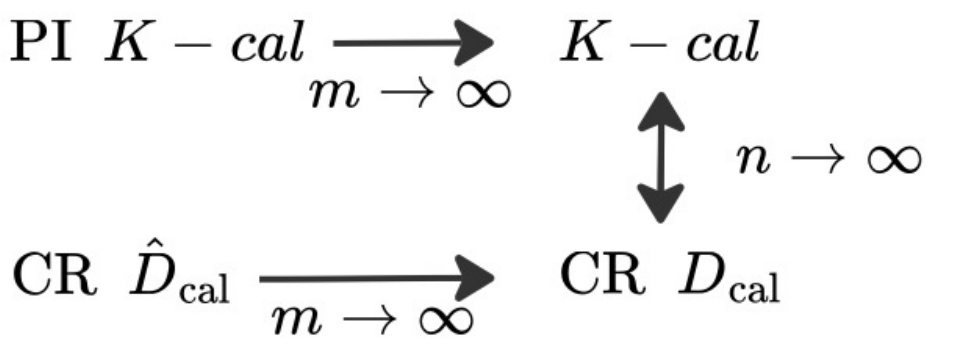}
    \caption{\textbf{Links between all calibrations.} $n$ represents the number of samples in the training set and $m$ represents the number of samples the calibration set. The equivalence between both measures is true with a mild assumption.}
    \label{fig:linkcals}
\end{figure}

\section{TESTING THE CALIBRATION}\label{sec:tests}
Here, we describe the process that we used to design, compute, and understand the limitations of the different tests.
\subsection{Designing the tests}
As proven before (Th. \ref{thm:bigth}), the event-specific bucket $B^k_{[0, .]}$ should follow a uniform distribution. With this, we can evaluate the CR-D-calibration using a \citep{kolmogorov}-\citep{smirnov} (KS) test (with our definition of CR–D–Calibration, taking $\alpha=\infty$ reduces the criterion to the KS statistic). An inspirational  alternative, in \citet{haider_eective_nodate}, the authors discretize $[0, 1]$ into a fixed number of buckets and apply a Pearson $\chi^2$ test. Because testing is conducted separately for each competing event $k$, we adjust for multiplicity across events (\emph{e.g.}, \cite{bonferroni}), acknowledging this yields a conservative decision rule. \\
In the same vein, for the PI-calibrations, we can assess the model's calibration with the KS test between the \citep{aalen_empirical_1978} estimator and the marginal predictions of the estimator. \\
More powerful tests could be defined by exploring the limit distribution, and/or integrating bootstrap techniques to better address the censoring distribution. 

\subsection{Computing the tests}
To compute the proposed calibration tests, we employ the same methodologie for CR-D-calibration and PI-${\rm cal}_{K}^\alpha$-calibration.

For both calibrations, we use a KS test. Here is the description on the scheme for the CR-D-calibration. This can be adapted to the PI-${\rm cal}_{K}^\alpha$-calibration, using the marginal predictions of the PI-estimator (in our case the \citet{aalen_empirical_1978} estimator) instead of the Uniform distribution. 

\begin{itemize}
    \item We predict the CIFs for each event on the $\mathcal{D}_{test}$.
    \item We compute the buckets from the CR-D)calibration. The buckets should follow a uniform distribution.
    \item We make a KS test onto each event to test the null hypothesis that they are drawn from a standard uniform distribution Uniform(0,1).
    \item Since a separate test is conducted for each competing event $k$, the resulting p-values must be adjusted for multiplicity (using the Bonferroni correction) to maintain the overall significance level across all tests. 
    \item The test is considered has passed if the p-values are higher than $0.05 \times K$, with $K$ the number of competing events.  
\end{itemize}

\subsection{Limitations}
For the tests, we want to acknowledge that the tests are not that powerful. Indeed, we see that most models may pass one or several tests while being insufficiently calibrated (as shown in Figure \ref{fig:meanincfunc}). We really want to recommend using the metrics instead of the tests to understand which model is more calibrated, etc. But, in a primary approach, the calibration tests can be applied to have a more precise idea.

\section{RECALIBRATION}
\subsection{Recalibration using AJ-k calibration: Framework}\label{sec:aj-recal}

We describe here our process to recalibrate competing risks models. The same process can be applied to survival models with only one event of interest.

We have $\mathcal{D} = (x_i, t_i, \delta_i)_{1 \leq i \leq N}$. 
Our dataset is split into three different sets: a training set $\mathcal{D}_{train}$, a testing set $\mathcal{D}_{test}$, and a third set: a calibration set $\D_{cal}$. \\
Our recalibration process can be explained as:
\begin{itemize}
    \item We train a given model onto $\mathcal{D}_{train}$. 
    \item At calibration time, we denote $\mathcal{D}_{cal} = (\mathcal{X}_{cal}, \mathcal{Y}_{cal})$. For each event (resp. survival to any event), we predict the incidence function for each individual $\mathbf{x}_i \in \mathcal{X}_{cal}$ at $d$ fixed times $(\tau_j)_{1\leq j \leq d}$. The different incidence functions (resp. survival to any event) for each individual $\hat{F}_k(.|\mathbf{X} = \mathbf{x}_i)$ (resp. $\hat{F}_0(.|\mathbf{X} = \mathbf{x}_i) \defeq \hat{S}(.|\mathbf{X} = \mathbf{x}_i)$) are called individualized incidence functions (IIFs) in the following. We choose the different times $(\tau_\nu)_{1\leq \nu \leq d}$ as the quantiles of the duration of the calibration set. We train an \citet{aalen_empirical_1978} estimator onto $\D_{cal}$ and predict the IIFs at the same times $(\tau_\nu)_{1\leq \nu \leq d}$.
    \item We compute the marginal differences between the \citet{aalen_empirical_1978} estimator and the model on $\D_{cal}$. \emph{I.e.} $\forall k \in [0, K],\forall \tau_j, \, d_{j}^k = \hat{F}_{k}^{AJ}(\tau_j) - \frac{1}{|\D_{cal}|} \sum_{\mathbf{x}\in \mathcal{X}_{conf}} \hat{F}_k(\tau_j |\mathbf{X} = \mathbf{x}_i)$ to approximate $\hat{F}^{AJ}_{k}(\tau_j) - \int_{\mathbf{x}\in \mathcal{X}_{cal}} \hat{F}_k(\tau_j| \mathbf{X}= \mathbf{x}) d\mathbf{x}$. 
    \item Then, on a new set of data $\mathcal{D}_{test}$, we predict the IIFs with the trained model at $(\tau_j)_{1\leq j \leq \nu}$ and obtain $(\hat{F}_k(\tau_j |\mathbf{X} = \mathbf{x}_i))_{1\leq j \leq \nu, \mathbf{x}_i \in \mathcal{X}_{test}, k \in \llbracket 0, K\rrbracket}$. 
    Then, for each $x_i$, we return the re-calibrated probabilities at each time $\tau_j$ defined as $$\forall k \in \llbracket 0, K \rrbracket, \tilde{F}_k(\tau_j| \mathbf{X}= \mathbf{x}_i) = \hat{F}_k(\tau_j| \mathbf{X}= \mathbf{x}_i) - d^k_{j}.$$
\end{itemize}
This framework is close to the one used in conformalized survival analysis \citet{qi_conformalized_2024, farina_doubly_2025}.

\subsection{Recalibration using temperature scaling}\label{sec:temp-scaling}
Here, we explain how we adapted the temperature scaling framework \cite{guo_calibration_2017} for recalibration in the competing risks setting, re-using the AJ-k recalibration framework.

This method works as follows:
\begin{itemize}
    \item We train a given model onto $\mathcal{D}_{cal}$
    \item We take the quantile of the duration distribution on $\mathcal{D}_{cal}$), that gives us a time grid on which we will perform the recalibration. This recalibration will be performed independently at each of these times.
    \item At each fixed time point $\tau$, a separate temperature scaling model is trained. This typically involves learning a single scalar parameter, $\beta(\tau)$ (the "temperature"), which is applied multiplicatively to the logits (or linear scores) of the original model's output before the final activation function (a softmax).
    \item The true or "actual proportion" of events at the chosen time point is estimated using a non-parametric method, such as the  \citep{aalen_empirical_1978} estimator or another form of marginally consistent PI estimator. This estimator provides the gold-standard, non-parametric survival probability at that specific time, serving as the calibration target for the temperature scaling model. 
    \item The temperature scaling model is trained to minimize the difference (\emph{e.g.}, using a cross-entropy loss) between the original model's marginal predictions and the PI-estimator's output for that same time $\tau$.
    \item To summarize, you are not scaling each of the cumulative incidence function with one temperature; you are finding an optimal temperature parameter $\beta(\tau)$ for each time $\tau$ to align the model's predicted marginal probability with the non-parametric marginal probability.
\end{itemize}

\section{EXPERIMENTAL PART}\label{sec:moreresults}
\subsection{Experimental details}
As said in the main part, we compare different state-of-the-art competing risks models. 
All models have been trained with their default parameters on an internal cluster (40 CPUs, 252 Gb of RAM, 4 NVIDIA Tesla V100 GPUs). The implementation was done in Python, using the lifelines library for the \citet{aalen_empirical_1978} estimator, the Pycox library \url{https://github.com/havakv/pycox} for the DeepHit model \citet{lee_deephit_2018}, the "cmprsk" from the R library using rpy2 to train the Random Survival Forests model \cite{ishwaran_random_2008} and the hazardous Python library \url{https://github.com/soda-inria/hazardous} to train SurvTRACE \cite{wang_survtrace_2022} and compute the C-index, and the IBS. The calibration metrics in survival analysis were computed thanks to the SurvivalEval Python library \url{https://github.com/shi-ang/SurvivalEVAL}. \\
Following \citet{candes_conformalized_2023}, to run the experiments, we have divided our dataset into three different dataset: 40\% was used for the training set ($\mathcal{D}_{train}$), 40\% to re-calibrate the models ($\mathcal{D}_{cal}$) and 20\% for the test set ($\mathcal{D}_{test}$).

\subsection{More results}\label{sec:resultsappendix}
Below are all the experimental results for the different datasets. For each dataset, we will show the calibration metrics, and then the IBS and the C-index before and after recalibration.

For the Table \ref{tab:seer100kdcrtest}, we see one limitation of the tests. Although the AJ-${\rm cal}_{K}^\alpha$-calibration is really small (around $10^{-4}$, Fig \ref{fig:metrics_seer100k}), not all tests are positive because of the sample size. A slight difference has an enormous impact on the test result. 

Tables \ref{tab:seer10kcalibperevent} and \ref{tab:seer100kcalibperevent} reveal that, with the exception of DeepHit, the minority class (event 2) of the SEER dataset (Fig. \ref{fig:distribdatasets}) exhibits the largest calibration error across all models. This heterogeneity suggests that current competing risks models generally struggle to maintain good calibration on less frequent events, which represents a significant potential issue.

\begin{figure}[h!]
    \begin{minipage}[b]{0.48\linewidth}
        \centering
    \includegraphics[width=1.05\textwidth]{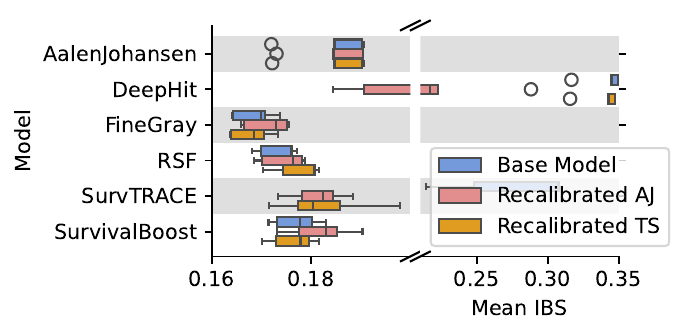}  
        \caption{Mean Integrated Brier Score}
    \end{minipage}
    \begin{minipage}[b]{0.48\linewidth}
        \centering
        \includegraphics[width=\textwidth]{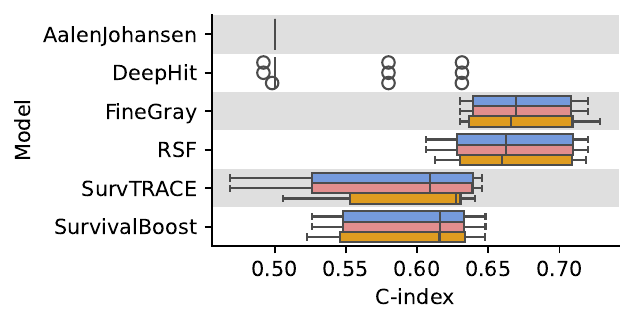}  
        \caption{C-index}
    \end{minipage}
    \caption{\textbf{Synthetic Dataset}, before and after recalibration, usual metrics.}
    \label{fig:metrics_synthe}
\end{figure}

\begin{figure}[h!]
    \begin{minipage}[b]{0.48\linewidth}
        \centering
    \includegraphics[width=\textwidth]{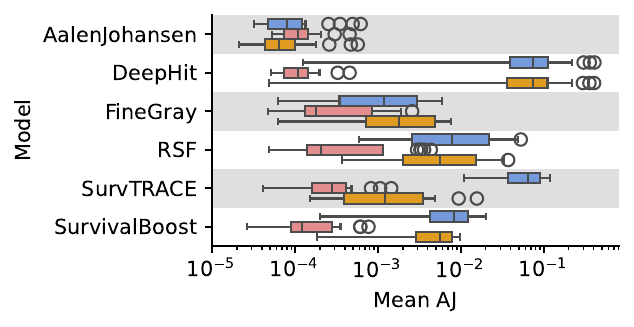}  
        \caption{AJ Calibration}
    \end{minipage}
    \begin{minipage}[b]{0.48\linewidth}
        \centering
        \includegraphics[width=\textwidth]{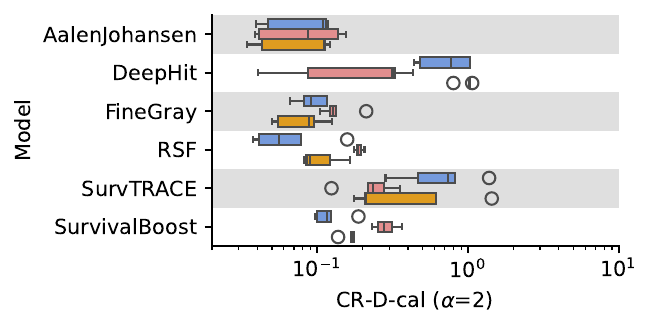}  
        \caption{CR-D-calibration}
    \end{minipage}
    \caption{\textbf{Synthetic Dataset}, before and after recalibration, calibration metrics.}
    \label{fig:calib_metrics_synthe}
\end{figure}

\begin{table}[h!]
    \centering
    \caption{\textbf{Synthetic Dataset} KS test of D-CR-calibration with the multi-test correction. Computations are made over 5 seeds and we show the percentage of tests that have passed. }
    \begin{tabular}{llrr}
    \toprule
    model & recalibration & Test D-CR-calibration & Test AJ-${\rm cal}_{K}^\alpha$-calibration  \\
    \midrule
    AalenJohansen & Base Model & 100\% & 100\% \\
    AalenJohansen & Recalibrated AJ & 60\%  & 100\%\\
    AalenJohansen & Recalibrated TS & 100\% & 100\% \\
    DeepHit & Base Model & 0\% & 20\% \\
    DeepHit & Recalibrated AJ & 40\% & 100\% \\
    DeepHit & Recalibrated TS & 0\%  & 20\%\\
    FineGray & Base Model & 80\%  & 100\%\\
    FineGray & Recalibrated AJ & 0\%  & 100\%\\
    FineGray & Recalibrated TS & 80\% & 100\% \\
    RSF & Base Model & 80\%  & 0\%\\
    RSF & Recalibrated AJ & 0\%  & 100\%\\
    RSF & Recalibrated TS & 80\%  & 0\%\\
    SurvTRACE & Base Model & 0\%  & 0\%\\
    SurvTRACE & Recalibrated AJ & 0\%  & 100\%\\
    SurvTRACE & Recalibrated TS & 0\%  & 60\%\\
    SurvivalBoost & Base Model & 80\%  & 0\%\\
    SurvivalBoost & Recalibrated AJ & 0\% & 100\% \\
    SurvivalBoost & Recalibrated TS & 20\% & 0\%\\
    \bottomrule
    \end{tabular}
    \label{tab:synthedcrtest}
\end{table}

\begin{table}[h!]
    \centering
    \caption{\textbf{Synthetic Dataset}: CR-$\hat{D}$-calibration per event.  We show the $\hat{D}_2^{CR}$. }
    \label{tab:synthecalibperevent}
    \begin{tabular}{llll}
    \toprule
    Model & Event 1 & Event 2 & Event 3 \\
    \midrule
    AalenJohansen & 0.03 ± 0.0 & 0.05 ± 0.03 & 0.04 ± 0.02 \\
    DeepHit & 0.31 ± 0.42 & 0.17 ± 0.22 & 0.2 ± 0.25 \\
    FineGray & 0.03 ± 0.01 & 0.05 ± 0.03 & 0.08 ± 0.02 \\
    RSF & 0.04 ± 0.01 & 0.05 ± 0.04 & 0.06 ± 0.03 \\
    SurvTRACE & 0.3 ± 0.14 & 0.34 ± 0.11 & 0.35 ± 0.13 \\
    SurvivalBoost & 0.09 ± 0.01 & 0.08 ± 0.04 & 0.06 ± 0.02 \\
    \bottomrule
    \end{tabular}
\end{table}

\vspace{10cm}

\begin{figure}[h!]
    \begin{minipage}[b]{0.48\linewidth}
        \centering
    \includegraphics[width=1.05\textwidth]{figs/metabric_ibs.pdf}  
        \caption{Mean Integrated Brier Score}
    \end{minipage}
    \begin{minipage}[b]{0.48\linewidth}
        \centering
        \includegraphics[width=\textwidth]{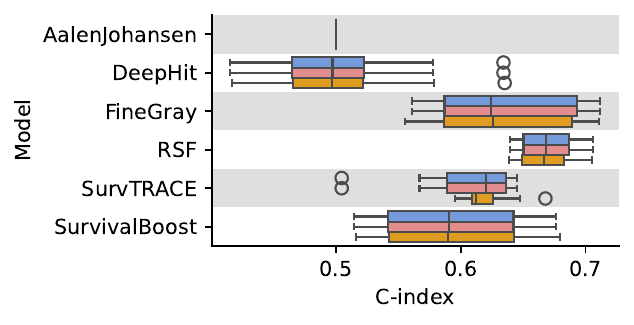}  
        \caption{C-index}
    \end{minipage}
    \caption{\textbf{METABRIC}, before and after recalibration, usual metrics.}
    \label{fig:metrics_metabric}
\end{figure}

\begin{figure}[h!]
    \begin{minipage}[b]{0.48\linewidth}
        \centering
    \includegraphics[width=\textwidth]{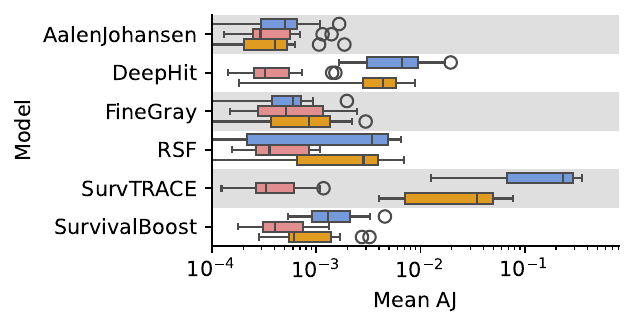}  
        \caption{AJ Calibration}
    \end{minipage}
    \begin{minipage}[b]{0.48\linewidth}
        \centering
        \includegraphics[width=\textwidth]{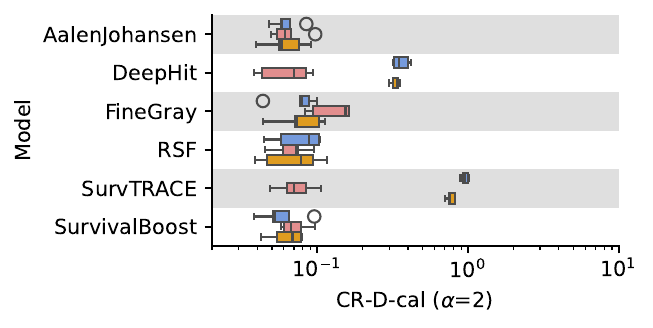}  
        \caption{CR-D-calibration}
    \end{minipage}
    \caption{\textbf{METABRIC}, before and after recalibration, calibration metrics.}
    \label{fig:calib_metrics_metabric}
\end{figure}

\begin{table}[h!]
    \centering
    \caption{\textbf{METABRIC} KS test of calibrations with the multi-test correction. Computations are made over 5 seeds and we show the percentage of tests that have passed. }
    \begin{tabular}{llrr}
    \toprule
    Model & Recalibration & Test CR-D-Calibration & Test AJ-${\rm cal}_{K}^\alpha$-calibration\\
    \midrule
    AalenJohansen & Base Model & 100\% & 100\% \\
    AalenJohansen & Recalibrated AJ & 100\% & 100\%\\
    AalenJohansen & Recalibrated TS & 100\% & 100\%\\
    DeepHit & Base Model & 0\% & 0\%\\
    DeepHit & Recalibrated AJ & 100\% & 100\%\\
    DeepHit & Recalibrated TS & 0\% & 80\%\\
    FineGray & Base Model & 100\% & 100\%\\
    FineGray & Recalibrated AJ & 20\% & 100\%\\
    FineGray & Recalibrated TS & 20\% & 100\%\\
    RSF & Base Model & 100\% & 0\%\\
    RSF & Recalibrated AJ & 60\% & 100\%\\
    RSF & Recalibrated TS & 100\% & 0\%\\
    SurvTRACE & Base Model & 0\% & 0\%\\
    SurvTRACE & Recalibrated AJ & 100\% & 100\%\\
    SurvTRACE & Recalibrated TS & 0\% & 0\%\\
    SurvivalBoost & Base Model & 100\% & 80\%\\
    SurvivalBoost & Recalibrated AJ & 60\% & 100\%\\
    SurvivalBoost & Recalibrated TS & 100\% & 80\%\\
    \bottomrule
    \end{tabular}
    \label{tab:metabrictestdcr}
\end{table}

\begin{table}[h!]
    \centering
    \caption{\textbf{METABRIC}: CR-$\hat{D}$-calibration per event.  We show the $\hat{D}_2^{CR}$.}
    \label{tab:metabriccalibperevent}
    \begin{tabular}{lll}
    \toprule
    Model & Event 1 & Event 2 \\
    \midrule
    AalenJohansen & 0.05 ± 0.02 & 0.05 ± 0.01 \\
    DeepHit & 0.34 ± 0.06 & 0.14 ± 0.03 \\
    FineGray & 0.05 ± 0.02 & 0.07 ± 0.02 \\
    RSF & 0.04 ± 0.01 & 0.08 ± 0.03 \\
    SurvTRACE & 0.58 ± 0.06 & 0.77 ± 0.04 \\
    SurvivalBoost & 0.05 ± 0.01 & 0.04 ± 0.02 \\
    \bottomrule
    \end{tabular}
\end{table}

\begin{figure}[h!]
    \begin{minipage}[b]{0.48\linewidth}
        \centering
    \includegraphics[width=1.05\textwidth]{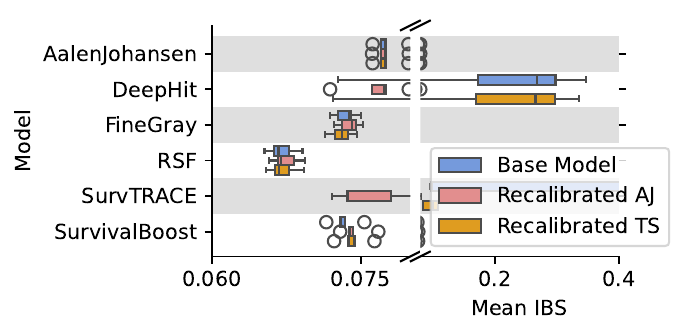}  
        \caption{Mean Integrated Brier Score}
    \end{minipage}
    \begin{minipage}[b]{0.48\linewidth}
        \centering
        \includegraphics[width=\textwidth]{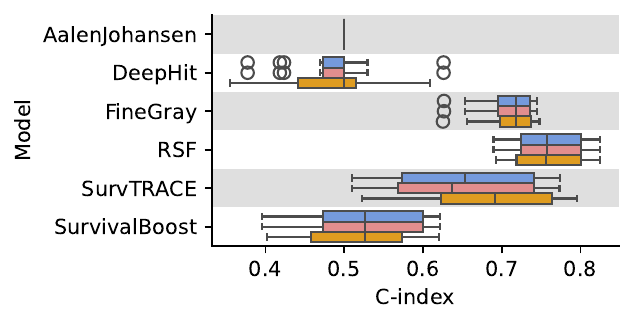}  
        \caption{C-index}
    \end{minipage}
    \caption{\textbf{SEER 10k training samples}, before and after recalibration, usual metrics.}
    \label{fig:metrics_seer10k}
\end{figure}

\begin{figure}[h!]
    \begin{minipage}[b]{0.48\linewidth}
        \centering
    \includegraphics[width=\textwidth]{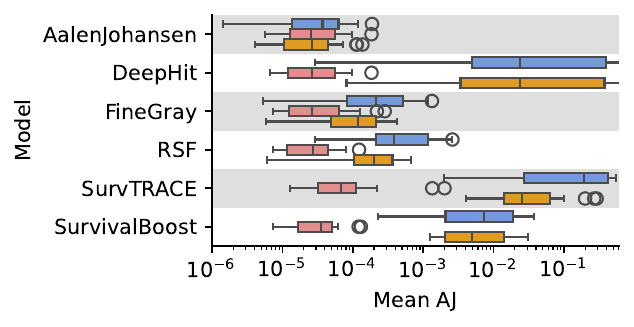}  
        \caption{AJ Calibration}
    \end{minipage}
    \begin{minipage}[b]{0.48\linewidth}
        \centering
        \includegraphics[width=\textwidth]{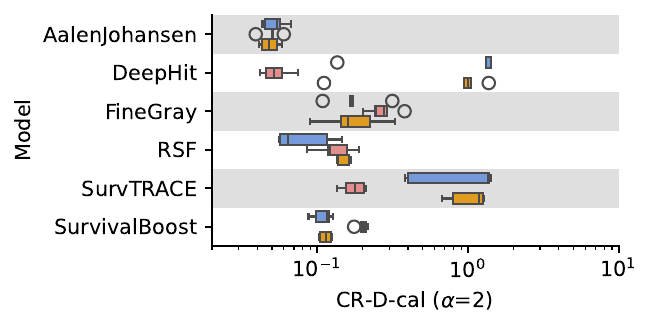}  
        \caption{CR-D-calibration}
    \end{minipage}
    \caption{\textbf{SEER 10k training samples}, before and after recalibration, calibration metrics.}
    \label{fig:seer10k_metrics_synthe}
\end{figure}

\begin{table}[ht!]
    \centering
    \caption{\textbf{SEER 10k training samples}, KS test of the calibrations with the multi-test correction. Computations are made over 5 seeds and we show the percentage of tests that have passed. }
    \begin{tabular}{llrr}
    \toprule
    Model & Recalibration & Test D-CR-Calibration & Test AJ-${\rm cal}_{K}^\alpha$-calibration \\ 
    \midrule
    AalenJohansen & Base Model & 100\% & 100\%\\
    AalenJohansen & Recalibrated AJ & 100\% & 100\%\\
    nAalenJohansen & Recalibrated TS & 100\% & 100\%\\
    DeepHit & Base Model & 0\% & 0\% \\
    DeepHit & Recalibrated AJ & 80\% & 100\%\\
    DeepHit & Recalibrated TS & 20\% & 0\%\\
    FineGray & Base Model & 20\% & 100\%\\
    FineGray & Recalibrated AJ & 0\% & 100\%\\
    FineGray & Recalibrated TS & 0\% & 100\%\\
    RSF & Base Model & 80\% & 40\%\\
    RSF & Recalibrated AJ & 0\% & 100\%\\
    RSF & Recalibrated TS & 0\% & 100\%\\
    SurvTRACE & Base Model & 0\% & 0\%\\
    SurvTRACE & Recalibrated AJ & 0\% & 100\%\\
    SurvTRACE & Recalibrated TS & 0\% & 0\%\\
    SurvivalBoost & Base Model & 100\% & 0\%\\
    SurvivalBoost & Recalibrated AJ & 0\% & 100\%\\
    SurvivalBoost & Recalibrated TS & 100\% & 0\%\\
    \bottomrule
    \end{tabular}
    \label{tab:seer10ktestdcal}
\end{table}

\begin{table}[h!]
    \centering
    \caption{\textbf{SEER 10k training points}: CR-$\hat{D}$-calibration per event.  We show the $\hat{D}_2^{CR}$.}
    \label{tab:seer10kcalibperevent}
    \begin{tabular}{llll}
    \toprule
    Model & Event 1 & Event 2 & Event 3 \\
    \midrule
    AalenJohansen & 0.03 ± 0.01 & 0.04 ± 0.01 & 0.03 ± 0.01 \\
    DeepHit & 0.77 ± 0.42 & 0.73 ± 0.38 & 0.79 ± 0.43 \\
    FineGray & 0.05 ± 0.03 & 0.15 ± 0.09 & 0.12 ± 0.06 \\
    RSF & 0.04 ± 0.01 & 0.08 ± 0.04 & 0.04 ± 0.01 \\
    SurvTRACE & 0.54 ± 0.28 & 0.61 ± 0.34 & 0.57 ± 0.31 \\
    SurvivalBoost & 0.07 ± 0.01 & 0.06 ± 0.02 & 0.07 ± 0.01\\
    \bottomrule
    \end{tabular}
\end{table}

\begin{figure}[h!]
    \begin{minipage}[b]{0.48\linewidth}
        \centering
    \includegraphics[width=1.05\textwidth]{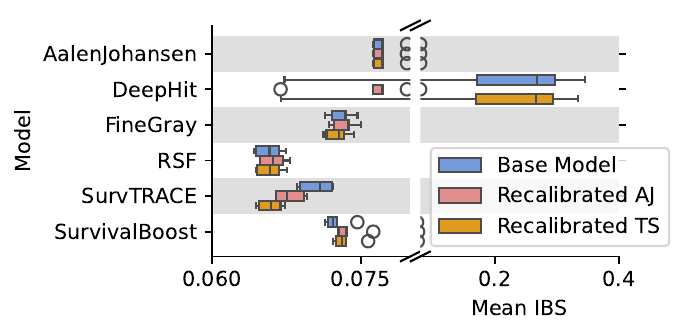}  
        \caption{Mean Integrated Brier Score}
    \end{minipage}
    \begin{minipage}[b]{0.48\linewidth}
        \centering
        \includegraphics[width=\textwidth]{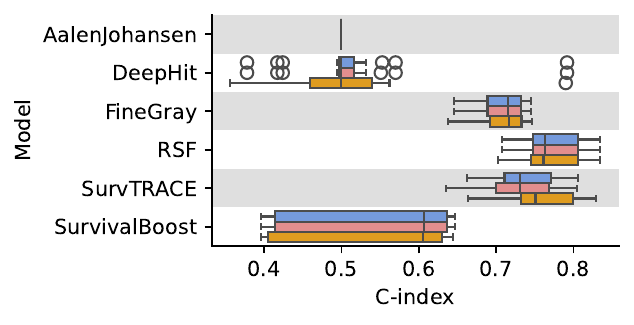}  
        \caption{C-index}
    \end{minipage}
    \caption{\textbf{SEER 100k training samples}, before and after recalibration, usual metrics.}
    \label{fig:metrics_seer100k}
\end{figure}

\begin{figure}[h!]
    \begin{minipage}[b]{0.48\linewidth}
        \centering
    \includegraphics[width=\textwidth]{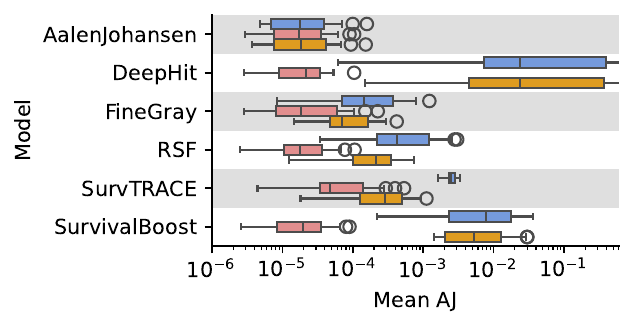}  
        \caption{AJ Calibration}
    \end{minipage}
    \begin{minipage}[b]{0.48\linewidth}
        \centering
        \includegraphics[width=\textwidth]{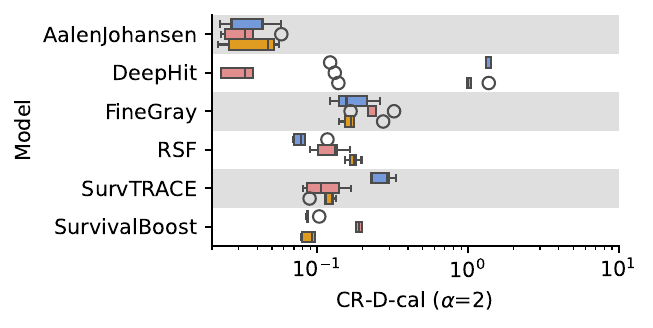}  
        \caption{CR-D-calibration}
    \end{minipage}
    \caption{\textbf{SEER 100k training samples}, before and after recalibration, calibration metrics.}
    \label{fig:seer100k_metrics_synthe}
\end{figure}

\begin{table}[h!]
    \centering
    \caption{\textbf{SEER 100k training samples}, KS test of the calibrations with the multi-test correction. Computations are made over 5 seeds and we show the percentage of tests that have passed. }
    \begin{tabular}{llrr}
    \toprule
    Model & recalibration & Test D-CR-Calibration & Test AJ-${\rm cal}_{K}^2$-calibration \\
    \midrule
    AalenJohansen & Base Model & 100\%& 100\% \\
    AalenJohansen & Recalibrated AJ & 100\% & 100\%\\
    AalenJohansen & Recalibrated TS & 100\% & 100\%\\
    DeepHit & Base Model & 0\% & 0\%\\
    DeepHit & Recalibrated AJ & 80\% & 100\%\\
    DeepHit & Recalibrated TS & 0\% & 0\%\\
    FineGray & Base Model & 20\% & 100\%\\
    FineGray & Recalibrated AJ & 0\% & 100\%\\
    FineGray & Recalibrated TS & 0.0 \% & 100\% \\
    RSF & Base Model & 80\% & 20\%\\
    RSF & Recalibrated AJ & 0\% & 100\%\\
    RSF & Recalibrated TS & 0\% & 60\%\\
    SurvTRACE & Base Model & 0\% & 40\%\\
    SurvTRACE & Recalibrated AJ & 0\% & 100\%\\
    SurvTRACE & Recalibrated TS & 20\% & 100\%\\
    SurvivalBoost & Base Model & 100\% & 0\%\\
    SurvivalBoost & Recalibrated AJ & 0\% & 100\%\\
    SurvivalBoost & Recalibrated TS & 100\% & 0\%\\
    \bottomrule
    \end{tabular}
    \label{tab:seer100kdcrtest}
\end{table}

\begin{table}[h!]
    \centering
    \caption{\textbf{SEER 100k training points}: CR-$\hat{D}$-calibration per event. We show the $\hat{D}_2^{CR}$.}
    \label{tab:seer100kcalibperevent}
    \begin{tabular}{llll}
    \toprule
    Model & Event 1 & Event 2 & Event 3 \\
    \midrule
    AalenJohansen & 0.02 ± 0.01 & 0.03 ± 0.01 & 0.02 ± 0.01 \\
    DeepHit & 0.77 ± 0.41 & 0.72 ± 0.39 & 0.79 ± 0.43 \\
    FineGray & 0.04 ± 0.02 & 0.14 ± 0.05 & 0.1 ± 0.04 \\
    RSF & 0.03 ± 0.0 & 0.08 ± 0.02 & 0.04 ± 0.02 \\
    SurvTRACE & 0.11 ± 0.03 & 0.19 ± 0.04 & 0.16 ± 0.02 \\
    SurvivalBoost & 0.04 ± 0.01 & 0.06 ± 0.01 & 0.05 ± 0.01 \\
    \bottomrule
    \end{tabular}
\end{table}

\FloatBarrier
\section{DATASETS DESCRIPTION} \label{sec:datasetsdescr}

For all datasets, the distributions of event types are summarized in Figure \ref{fig:hist_data}.

For the synthetic dataset, we consider a competing-risks setting with three causes and covariates $X = (\Lambda_1, \Lambda_3, S_1, S_2, S_3)$ with independent components drawn from uniform distributions:
\begin{itemize}
    \item $\Lambda_1 \sim \mathcal{U}(0.4, 0.9)$,
    \item $\Lambda_2 = 1$ (constant),
    \item $\Lambda_3 \sim \mathcal{U}(1.2, 3)$,
    \item $S_1 \sim \mathcal{U}(1, 20)$,
    \item $S_2 \sim \mathcal{U}(1, 10)$ and
    \item $S_3 \sim \mathcal{U}(1.5, 5)$.
\end{itemize}
Conditionally on $X$, we generate three event times $T_1, T_2, T_3$ and a censoring time $C$ as independent positive random variables with Weibull distributions with
\begin{itemize}
    \item $\mathbb{P}(T_k \le t \mid X)= 1 - \exp\bigl(- (t / \Lambda_k)^{S_k}\bigr)$ for $k = 1,2,3$, and
    \item $\mathbb{P}(C \le t \mid X)= 1 - \exp\bigl(- t / \Lambda_0\bigr)$ with $\Lambda_0 = 1.5 \cdot \mathbb{E}[T^*]$,
\end{itemize}
where $T^* = \min(T_1, T_2, T_3)$ and $\Delta^* = \arg\min_k T_k$.

\begin{figure}[h!]
    \begin{minipage}[b]{0.33\linewidth}
        \centering
    \includegraphics[width=\textwidth]{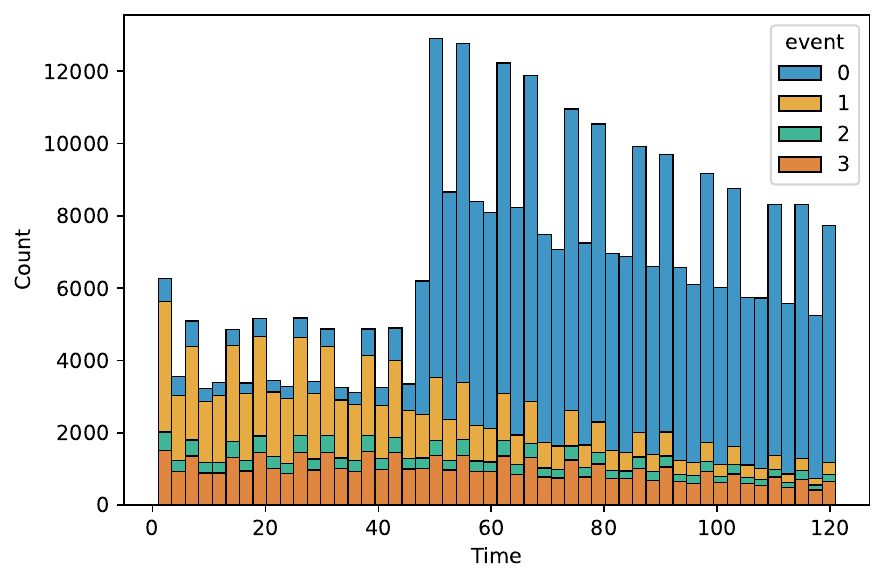}  
        \caption{SEER}\label{fig:hist_seer}
    \end{minipage}
    \begin{minipage}[b]{0.33\linewidth}
        \centering
        \includegraphics[width=\textwidth]{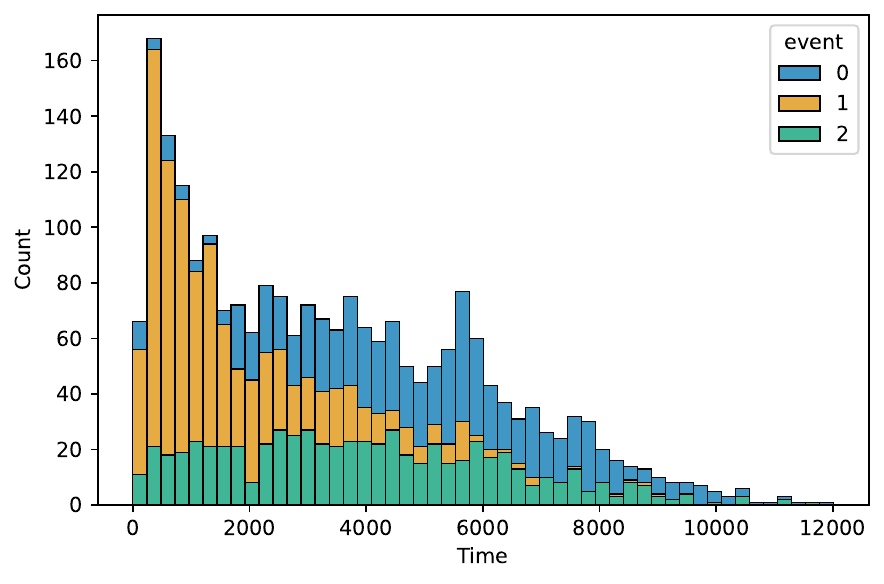}  
        \caption{METABRIC}\label{fig:hist_metabric}
    \end{minipage}
    \begin{minipage}[b]{0.33\linewidth}
        \centering
        \includegraphics[width=\textwidth]{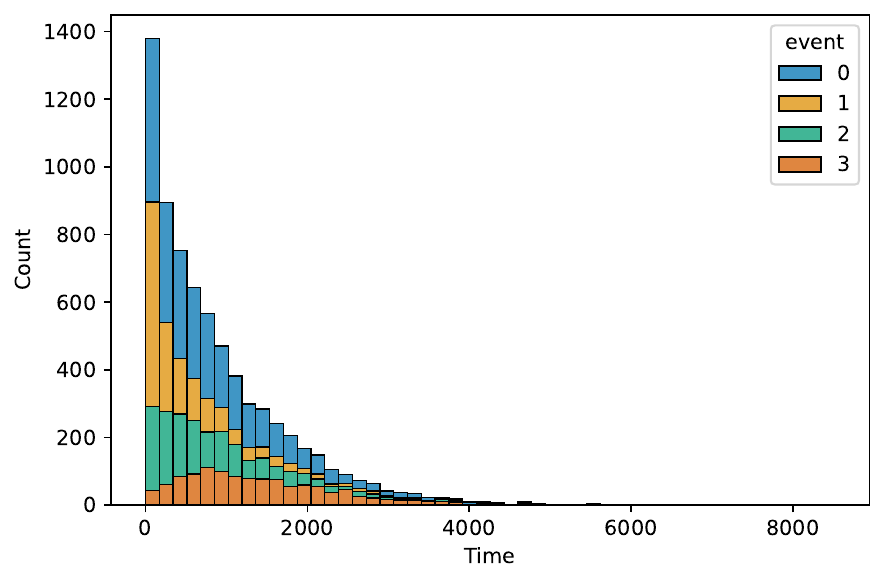}  
        \caption{Synthetic Dataset}\label{fig:hist_synthetic}
    \end{minipage}
    \caption{Histograms of the distribution of the events for each dataset}
    \label{fig:distribdatasets}\label{fig:hist_data}
\end{figure}

\FloatBarrier
\section{NOTATIONS}
\begin{table}[h!]
\caption{Notations used}
\label{tab:history}
\begin{tabularx}{\linewidth}{r r X}
\toprule
Maths Symbol & Domain & Description\\
\midrule 
    $K$ & $\mathbb{N}^*$ & number of competing events (events of interest)\\ 
    $\mathbf{X}$ & $\mathbb{R}^d$& the $d$ covariates the individual \\
    $T^*_k$ &$\mathbb{R}_+$& random variable of the time-to-event for event $k$\\
    $C$ &$\mathbb{R}_+$& random variable of the time-to-censoring\\
    $T^*$ & $\mathbb{R}_+$ & $\min (T^*_1, T^*_2, ..., T^*_K)$\\
    $T$ & $\mathbb{R}_+$ & $\min (T, C)$ \\
    $\Delta^*$ &$[1, K]$&$\argmin\limits_{k \in [1, K]} (T^*_k)$\\
    $\Delta$ &$[0, K]$&$\argmin (C, T^*_1, T^*_2, ..., T^*_K)$\\
    $\mu^\star$ & & Distribution of ($\mathbf{X}, (T^\star, \Delta^\star))$)\\
    $\mu$ & & Distribution of ($\mathbf{X}, (T, \Delta))$)\\
    \midrule
    S & $\mathbb{R}^{d}\longrightarrow \mathbb{R}_+$ & Survival function\\
    F & $\mathbb{R}^{d}\longrightarrow \mathbb{R}_+$ & Cumulative Incidence Function \\
    \midrule
    $\mathcal{D}$ & & Dataset considered \\
    $n$ &$\mathbb{N}^*$& number of individuals in the training set \\
    $m$ &$\mathbb{N}^*$& number of individuals in the calibration set \\
    $\mathbf{x}$ & & individuals observed \\
    $t$ &$\mathbb{R}_+^n$& time-to-event observed\\
    $\delta$ & $[0, K]$& event observed, 0 indicates censoring \\
    $B$ & $\mathbb{N}^*$ & Number of buckets \\
    \midrule
    $I_k \defeq [l_k, l_{k+1}]$ & & $k^{th}$ interval \\
    $F_T \defeq F_k(T|X)$ & $\mathbb{R}_+$& \\
    $F_{./\infty}  \defeq F_k(.|X)$ &$\mathbb{R}_+$ & \\
    PI & & Plug-in estimator \\
\bottomrule
\end{tabularx}
\end{table}

\end{document}